\documentclass[aps, twocolumn, superscriptaddress, 10pt]{revtex4-2}

\usepackage{amsmath, amssymb, amsthm}
\usepackage{bm}
\usepackage{braket}
\usepackage[normalem]{ulem}
\usepackage{physics}
\usepackage{quantikz,tikz}
\usepackage{color}
\usepackage[caption=false]{subfig}
\usepackage{orcidlink}
\usepackage{epsfig}

\usepackage{booktabs}
\usepackage{multirow}
\usepackage{makecell}

\usepackage{todonotes}

\newcommand{\trnorm}[1]{\norm*{#1}_{\mathrm{tr}}}
\newcommand{\trnormBig}[1]{\norm\Big{#1}_{\mathrm{tr}}}
\newcommand{\trnormBigg}[1]{\norm\Bigg{#1}_{\mathrm{tr}}}
\newcommand{\spec}{\operatorname{spec}}

\DeclareMathOperator{\cnot}{\mathrm{CNOT}}

\DeclareMathOperator{\crz}{\mathrm{CRZ}}

\theoremstyle{plain}
\newtheorem{theorem}{Theorem}
\newtheorem{proposition}[theorem]{Proposition}
\newtheorem{lemma}[theorem]{Lemma}

\newtheorem*{proposition*}{Proposition}
\newtheorem*{theorem*}{Theorem}
\newtheorem*{lemma*}{Lemma}

\theoremstyle{definition}

\theoremstyle{remark}
\newtheorem*{remark}{Remark}

\begin{document}
\title{Hamiltonian simulation for 3D elastic wave equations in homogeneous elastic media}

\author{Kosuke Nakanishi}
\email[Contact author]{Kosuke.Nakanishi.kd@mosk.tytlabs.co.jp}
\affiliation{Toyota Central R\&D Labs., Inc., 1-41, Yokomichi, Nagakute, Aichi 480-1192, Japan}

\author{Hiroshi Yano}
\affiliation{Toyota Central R\&D Labs., Inc., 1-41, Yokomichi, Nagakute, Aichi 480-1192, Japan}

\author{Yuki Sato}
\affiliation{Toyota Central R\&D Labs., Inc., 1-41, Yokomichi, Nagakute, Aichi 480-1192, Japan}

\begin{abstract}
We present an explicit quantum circuit construction for Hamiltonian simulation of a first-order velocity--stress formulation of the three-dimensional elastic wave equation in homogeneous isotropic media. Previous studies have shown how elastic wave equations can be cast into forms amenable to Hamiltonian simulation, but they typically rely on black box Hamiltonian access assumptions, making gate complexity estimation difficult.
Starting from the first-order velocity--stress formulation, we discretize the system by finite differences, transform it into Schr\"odinger form, and exploit the separation between the component register and the spatial register to decompose the Hamiltonian into structured tensor product terms. This yields explicit implementations of first-order and second-order Trotter formulas for the resulting time evolution operator. We derive corresponding error bounds and constant sensitive qubit and CNOT complexity estimates in terms of the discretization parameter, simulation time, target accuracy, and material parameters. Numerical experiments validate the proposed framework through comparisons with the exact time evolution and reconstructed physical fields.
\end{abstract}

\maketitle

\section{Introduction}
Partial differential equations (PDEs) are fundamental mathematical models for a wide range of physical phenomena, including heat conduction, fluid flow, electromagnetic wave propagation, quantum dynamics, and elastic deformation.
Solving PDEs is therefore essential for understanding and predicting these phenomena~\cite{PartialDifferentialEquationsBasicTheory2023a}.
Among such PDEs, a representative example is the elastic wave equation, which describes the propagation of stress and strain in elastic media~\cite{Memoireloismouvementfluides1822}.
In three dimensions, its first-order formulation has nine components per spatial point, and this equation is also used in seismology and volcanic wave modeling~\cite{QuantitativeSeismology2002}
Numerical simulation of PDE dynamics typically starts with spatial discretization, using finite difference method, finite element, or spectral methods~\cite{Accuracyfinitedifferencefiniteelementmodelingscalarelasticwaveequations1984,1FiniteDifferenceApproximations2007}.
After discretization, the resulting large scale systems become increasingly expensive to solve on classical architectures as problem size grows~\cite{reviewspectralpseudospectralfinitedifferencefiniteelementmodellingtechniquesgeophysicalimaging2011}. In response to these computational challenges, quantum computing has attracted increasing interest as a potential approach to accelerating the simulation of differential equations and PDE governed physical systems~\cite{applicationsquantumcomputingplasmasimulations2021,Improvedquantumalgorithmslinearnonlineardifferentialequations2023,Quantumalgorithmpartialdifferentialequationsnonconservativesystemsspatiallyvaryingparameters2025,QuantumcomputingfluiddynamicsusinghydrodynamicSchrodingerequation2023,Timecomplexityanalysisquantumalgorithmslinearrepresentationsnonlinearordinarypartialdifferentialequations2023}. This motivates the study of quantum algorithms for PDE simulation.

Quantum algorithms for PDE simulation can be broadly categorized into two main strategies.
The first approach discretizes both space and time and reduces the PDE to linear systems, which are then solved using quantum linear system solver algorithms, such as Harrow--Hassidim--Lloyd (HHL) algorithms~\cite{QuantumAlgorithmLinearSystemsEquations2009}.
The second strategy discretizes spatial variables to reduce PDEs to ordinary differential equations (ODEs), and then simulates their time evolution within a Hamiltonian simulation framework.
Representative approaches in this class include Schr\"{o}dingerization~\cite{QuantumSimulationPartialDifferentialEquationsSchrodingerization2024,QuantumsimulationpartialdifferentialequationsApplicationsdetailedanalysis2023} and linear combination of Hamiltonian simulation (LCHS)~\cite{LinearcombinationHamiltoniansimulationnonunitarydynamicsoptimalstatepreparationcost2023}.
In addition to these general frameworks, structure exploiting quantum algorithms have also been developed for specific equations, such as the advection equation~\cite{Hamiltoniansimulationhyperbolicpartialdifferentialequationsscalablequantumcircuits2024}, the wave equation~\cite{Quantumalgorithmsimulatingwaveequation2019}, coupled oscillator systems~\cite{ExponentialQuantumSpeedupSimulatingCoupledClassicalOscillators2023}, and the elastic wave equations~\cite{Quantumwavesimulationsourceslossfunctions2025,quantumcomputingconcept1Delasticwavesimulationexponentialspeedup2024}.

Gate complexity estimates have been reported for several PDEs, including the advection and wave equations~\cite{Hamiltoniansimulationhyperbolicpartialdifferentialequationsscalablequantumcircuits2024}, the Black--Scholes equation~\cite{QuantumCircuitsBlackScholesequationsSchrodingerisation2025a}, and Maxwell's equations with time dependent source terms~\cite{SchrodingerizationbasedQuantumCircuitsMaxwellsEquationtimedependentsourceterms2024a}.
For the elastic wave equation, prior studies, including Schr\"odingerization-based methods and direct Schr\"odinger mapping methods, have established practical formulations for Hamiltonian simulation~\cite{QuantumsimulationelasticwaveequationsSchrodingerisation2025a,Quantumwavesimulationsourceslossfunctions2025,quantumcomputingconcept1Delasticwavesimulationexponentialspeedup2024}.
However, their complexity analyses typically assume efficient black box Hamiltonian-access oracles for implementing $e^{-i\mathcal{H}t}$, which makes it difficult to obtain explicit gate complexity estimate for explicit circuits.
Thus, explicit gate complexity estimates for the three-dimensional elastic wave equation remain limited.

In this work, we present an explicit quantum circuit construction for Hamiltonian simulation of the first-order velocity--stress formulation of the three-dimensional elastic wave equation in homogeneous isotropic media, together with constant-sensitive CNOT cost estimates for both first-order and second-order Trotter schemes.
The three-dimensional elastic wave system is vector valued, with a nine component state at each spatial point.
Our implementation separates the state register (velocity and stress components) from the spatial register (grid indices).
On the spatial register, we implement the evolution operator generated by difference operators following Sato et al.~\cite{Hamiltoniansimulationhyperbolicpartialdifferentialequationsscalablequantumcircuits2024}, while operators on the component space are constructed through an eigendecomposition based procedure tailored to isotropic media.
We also report numerical experiments to validate the proposed simulation framework.

The remainder of this paper is organized as follows.
Section \ref{sec:preliminaries} introduces the mapping from the elastic wave equation to the Schr\"{o}dinger equation.
Section \ref{sec:quantum-circuit-implementation} presents the explicit circuit construction for the mapped Schr\"{o}dinger dynamics and analyzes gate complexity.
Section \ref{sec:numerical-experiments} reports numerical experiments validating the proposed method.
Finally, Section \ref{sec:conclusion} concludes the paper with limitations and future directions.

\section{Preliminaries}\label{sec:preliminaries}
In this section, we first introduce the elastic wave equation and rewrite it as a first-order system in time.
We then discretize the resulting system by a finite difference method and reformulate it into the Schr\"{o}dinger form following B\"osch et al.~\cite{Quantumwavesimulationsourceslossfunctions2025}.

The elastic wave equation governs the time evolution of displacement and stress in elastic media.
Let $\Omega \subset \mathbb{R}^3$ be a bounded domain with boundary $\partial \Omega$ and let time $T>0$.
For $(\bm{x},t)\in \Omega\times[0,T]$, let $\bm{u}(\bm{x},t)$ and $\bm{\sigma}(\bm{x},t)$ denote the displacement field and Cauchy stress tensor, respectively.
The equation of motion is
\begin{align}\label{eq:ew-orig}
	\rho \frac{\partial^2 \bm{u}(\bm{x},t)}{\partial t^2} &= \nabla \cdot \bm{\sigma}(\bm{x}, t),
\end{align}
where $\rho$ is the material density, and the constitutive relation gives the relation:
\begin{align}
    \bm{\sigma}(\bm{x}, t) = S_{\mathrm{comp}} \nabla \bm{u}(\bm{x},t),
\end{align}
with the compliance tensor $S_{\mathrm{comp}}$.

To rewrite the elastic wave equation in first-order form, we introduce the velocity field and Voigt form.
Define the velocity field as $\bm{v}(\bm{x},t):=\partial_{t}\bm{u}(\bm{x},t)\in \mathbb{R}^3$, and represent the Cauchy stress tensor in Voigt form by
\begin{align}
	\mathrm{vec}[\bm{\sigma}](\bm{x},t):=
	[\sigma_{xx};\sigma_{yy};\sigma_{zz};\sigma_{xy};\sigma_{xz};\sigma_{yz}](\bm{x},t)\in\mathbb{R}^6.
\end{align}
Accordingly, we define
\begin{align}
	\bm{w}(\bm{x},t) := [\bm{v};\mathrm{vec}[\bm{\sigma}]]^{\mathrm{T}}(\bm{x},t)\in \mathbb{R}^9. \label{eq:w_vec}
\end{align}
Then the initial boundary value problem is written as
\begin{align}\label{eq:ew-base}
	\dv{t}\bm{w}(\bm{x},t) &= \mathcal{B}^{-1}\mathcal{A}\bm{w}(\bm{x},t), \quad (\bm{x},t)\in\Omega\times[0,T],\\
	\bm{w}(\bm{x},0) &= \bm{w}_0(\bm{x}),\\
	\bm{w}(\bm{x},t) &= \bm{0}, \quad \bm{x}\in\partial\Omega.
\end{align}
Here, $\bm{w}_0(\bm{x})$ is an initial state.
The operators $\mathcal{A}$ and $\mathcal{B}$ are defined by
\begin{align}
	\mathcal{A}:=
	\begin{bmatrix}
		\mathcal{O}_{3\times3} & \mathcal{D} \\
		\mathcal{D}^{\mathrm T} & \mathcal{O}_{6\times6}
	\end{bmatrix},\quad
	\mathcal{D}:=
	\begin{bmatrix}
		\partial_x & 0 & 0 & \partial_y & \partial_z & 0\\
		0 & \partial_y & 0 & \partial_x & 0 & \partial_z\\
		0 & 0 & \partial_z & 0 & \partial_x & \partial_y
	\end{bmatrix}, \label{eq:A-D}
\end{align}
and
\begin{align}
	\mathcal{B}:=
	\begin{bmatrix}
		\rho I_{3\times3} & \mathcal{O}_{3\times6}\\
		\mathcal{O}_{6\times3} & \mathcal{S}_{\mathrm{comp}}
	\end{bmatrix}. \label{eq:B-continuous}
\end{align}
with $\partial_k$ ($k=x,y,z$) denoting the spatial derivative, and $\mathcal{O}_{m\times n}$ denoting the $m\times n$ zero operator matrix.
Note that the boundary condition imposed in Eq.~\eqref{eq:ew-base} is introduced for the auxiliary first-order velocity--stress formulation considered in this paper and does not coincide with the boundary condition of the original second-order displacement equation.
This simplification is adopted to keep the discussion focused on the circuit implementation of the core dynamical part of the algorithm.
In the semidiscrete setting, boundary conditions are tied to the concrete realization of the difference operators. In this direction, Sato \textit{et al.} discussed several such realizations for hyperbolic PDEs, including Dirichlet, Neumann, and periodic treatments~\cite{Hamiltoniansimulationhyperbolicpartialdifferentialequationsscalablequantumcircuits2024}.
Accordingly, the scope of the present work is limited to the analysis and implementation of the semidiscrete first-order system under this specific boundary treatment. The relation to physically natural boundary conditions for the original elastic wave equation, as well as extensions to other boundary realizations, will be considered in future works.

For the circuit construction in this paper, we assume homogeneous material parameters; the mass density $\rho$ and the compliance matrix $\mathcal{S}_{\mathrm{comp}}$ are independent of the spatial coordinate $\bm{x}$ and time $t$.
This homogeneity assumption simplifies the operator structure used in the following discretization and decomposition.

To construct a quantum representation, we discretize the first-order system on a uniform grid.
For this finite difference method, we restrict the computational domain to the cubic domain $\Omega=[0,L]^3\subset\mathbb{R}^3$, and divide each spatial axis into $2^n$ grid intervals.
Let $x_j:=jh$ for $j\in \{0,\ldots,2^n-1\}$ with $h:=L/(2^n-1)$ (and similarly for $y,z$).
For quantum encoding, we embed the state vector $\bm{w}(\bm{x},t)$ in Eq.~\eqref{eq:w_vec} into a 16-dimensional vector (4 qubits) by zero padding:
\begin{align}
	\bm{w}(\bm{x},t):=[\bm{v};\mathrm{vec}[\bm{\sigma}];\bm{0}_{1\times 7}]^{\mathrm T}(\bm{x},t)\in\mathbb{R}^{16}.
\end{align}
The resulting (unnormalized) quantum state is
\begin{align}
\ket*{\bm{w}_h(t)}
:=\sum_{k=0}^{15}\sum_{j_x,j_y,j_z=0}^{2^n-1}
w_{k,j_x,j_y,j_z}(t)\ket*{k}\ket*{j_x}\ket*{j_y}\ket*{j_z},
\end{align}
where $w_{k,j_x,j_y,j_z}(t)$ is the $k$-th component of $w_{j_x,j_y,j_z}(t):=\bm{w}(x_{j_x},y_{j_y},z_{j_z},t)\in\mathbb{R}^{16}$.
The sets $\{\ket*{k}\}_{k=0}^{15}$ and
$\{\ket*{j_x}\ket*{j_y}\ket*{j_z}\}_{j_x,j_y,j_z=0}^{2^n-1}$ form computational bases, which we refer to as the state register and the spatial register, respectively.

As the fundamental building block of the spatial discretization, we define a one-dimensional central difference operator on a uniform grid.
Let $q=(q_0,\dots,q_{N-1})^{\mathsf T}\in\mathbb{R}^N,\,N:=2^n$ be an arbitrary grid vector, and define $D_{\mathrm{cell}}\in\mathbb{R}^{N\times N}$ by
\begin{align}
	(D_{\mathrm{cell}}q)_j:=\frac{q_{j+1}-q_{j-1}}{2h}, 
\end{align}
for $j=0,\dots,N-1$ and $q_{-1}=q_N=0$.
Under the ghost-point convention adopted here, $D_{\mathrel{cell}}$ is anti-Hermitian and admits a matrix product operator (MPO) representation with bond dimension $2$~\cite{Tensornetworkreducedordermodelswallboundedflows2023}.
For the present first derivative operator, the MPO form is written as
\begin{align}
	D_{\mathrm{cell}}=\frac{1}{2h}\sum_{k=1}^{n}S_k^{\mathrm{cell}},
\end{align}
where
\begin{align}
	S_k^{\mathrm{cell}}
	&:=I^{\otimes(n-k)}\otimes \sigma_{01}\otimes \sigma_{10}^{\otimes(k-1)} \notag\\
	&\qquad -I^{\otimes(n-k)}\otimes \sigma_{10}\otimes \sigma_{01}^{\otimes(k-1)}.
\end{align}
with
\begin{align}
	\sigma_{01}:=
	\begin{bmatrix}
	0 & 1\\
	0 & 0
	\end{bmatrix},
	\qquad
	\sigma_{10}:=
	\begin{bmatrix}
	0 & 0\\
	1 & 0
	\end{bmatrix}.
\end{align}
This is a specific instance of a more general fact: local finite difference operators admit low bond dimension MPO representations.

We now lift the one-dimensional core operator to three dimensions in an axiswise manner.
For $\alpha\in\{1,2,3\}$ (corresponding to $x,y,z$), define
\begin{align}
	S_k^{(\alpha)}
	&:=I^{\otimes n(\alpha-1)}\otimes S_k^{\mathrm{cell}}\otimes I^{\otimes n(3-\alpha)},\\
	D^{(\alpha)}
	&:=\frac{1}{2h}\sum_{k=1}^{n}S_k^{(\alpha)}.
\end{align}
By construction, $D^{(\alpha)}$ acts only along axis $\alpha$.
For example for $\alpha=1$,
\begin{align}
	(D^{(1)}\bm{w})_{j_x,j_y,j_z}
	=
	\frac{\bm{w}_{j_x+1,j_y,j_z}-\bm{w}_{j_x-1,j_y,j_z}}{2h}.
\end{align}
The same ghost-point convention is applied on each axis.

Using $D^{(\alpha)}$ and $S_k^{(\alpha)}$, we define the discretized operators $A$ and $B$, completing the spatial discretization and leading to a finite dimensional ODE system.
Consistent with the 16-dimensional embedding of $\bm{w}(\bm{x},t)$, $\mathcal{A}$ and $\mathcal{B}$ are represented on the same dimensional state space.
The discrete counterpart of $\mathcal{A}$ is
\begin{align}\label{eq:A}
	A
	:=\sum_{\alpha=1}^{3}A^{(\alpha)}\otimes D^{(\alpha)}
	=\sum_{\alpha=1}^{3}\sum_{k=1}^{n}A^{(\alpha)}\otimes \frac{1}{2h}S_k^{(\alpha)},
\end{align}
where
\begin{align}
	A^{(\alpha)}:=
	\begin{bmatrix}
	O_{3\times3} & C^{(\alpha)} & O_{3\times7}\\
	(C^{(\alpha)})^{\mathsf T} & O_{6\times6} & O_{6\times7}\\
	O_{7\times3} & O_{7\times6} & O_{7\times7}
	\end{bmatrix},\label{eq:A-definition}
\end{align}
with
\begin{align}
C^{(1)}&:=
\begin{pmatrix}
1&0&0&0&0&0\\
0&0&0&1&0&0\\
0&0&0&0&1&0
\end{pmatrix},\\
C^{(2)}&:=
\begin{pmatrix}
0&0&0&1&0&0\\
0&1&0&0&0&0\\
0&0&0&0&0&1
\end{pmatrix},\\
C^{(3)}&:=
\begin{pmatrix}
0&0&0&0&1&0\\
0&0&0&0&0&1\\
0&0&1&0&0&0
\end{pmatrix}.\label{eq:c-matrix}
\end{align}
By construction, $A$ is anti-Hermitian.
Since $\mathcal{B}$ is position independent, its discrete counterpart is
\begin{align}\label{eq:B}
B
:=B_{\mathrm{cell}}\otimes I^{\otimes 3n},\qquad
B_{\mathrm{cell}}:=
\begin{bmatrix}
\rho I_{3\times3} & O_{3\times6} & O_{3\times7}\\
O_{6\times3} & S_{\mathrm{comp}} & O_{6\times7}\\
O_{7\times3} & O_{7\times6} & I_{7\times7}
\end{bmatrix}.
\end{align}
For physically stable materials, where $\rho>0$ and $S_{\mathrm{comp}}$, we have $B_{\mathrm{cell}}>0$, and thus $B>0$.
Therefore, the resulting finite dimensional ODE system is
\begin{align}
\dv{t}\ket*{\bm{w}_h(t)}=B^{-1}A\ket*{\bm{w}_h(t)}.\label{eq:discretize-ew}
\end{align}

The finite dimensional dynamics is now cast into the Schr\"{o}dinger form.
Starting from Eq.~\eqref{eq:discretize-ew}, define the transformed (generally unnormalized) state
\begin{align}
	\ket*{\tilde{\bm{u}}_h(t)}:=B^{1/2}\ket*{\bm{w}_h(t)}.
\end{align}
Multiplying Eq.~\eqref{eq:discretize-ew} from the left by $iB^{-1/2}$ yields
\begin{align}\label{eq:discrize-ew-schr}
	i\dv{t}\ket*{\tilde{\bm{u}}_h(t)}
	=
	\mathcal{H}\ket*{\tilde{\bm{u}}_h(t)},
\end{align}
where
\begin{align}\label{eq:H-aba}
	\mathcal{H}:=iB^{-1/2}AB^{-1/2}.
\end{align}
Since $A$ is anti-Hermitian and $B$ is positive definite, $\mathcal{H}$ is Hermitian.
The normalized state is defined by
\begin{align}
	\ket*{\bm{u}_h(t)}
	:=
	\frac{\ket*{\tilde{\bm{u}}_h(t)}}{\norm*{\ket*{\tilde{\bm{u}}_h(t)}}},
\end{align}
with $\norm*{\cdot}$ denoting $l^2$ norm.
\section{Quantum Circuit Implementation}\label{sec:quantum-circuit-implementation}
To perform the time evolution of Eq. (\ref{eq:discrize-ew-schr}), we construct an explicit quantum circuit based on the Trotter decomposition framework.
While advanced frameworks such as quantum signal processing (QSP) and quantum singular value transformation (QSVT) provide stronger asymptotic guarantees~\cite{HamiltonianSimulationQubitization2019,OptimalHamiltonianSimulationQuantumSignalProcessing2017,Quantumsingularvaluetransformationexponentialimprovementsquantummatrixarithmetics2019}, we prioritize circuit realization and constant sensitive gate counting in the present setting, and therefore adopt Trotter formulas as the Hamiltonian-simulation subroutine.

\subsection{Quantum Circuit for time evolution by Trotter decomposition}
This subsection focuses on one step ($\tau$ scale) simulation of the time evolution operator $\exp(-i\mathcal{H}\tau)$.
We construct first-order and second-order Trotter approximation and implement the corresponding quantum circuits.

To obtain a termwise decomposition of $\mathcal{H}$, we apply eigendecomposition to each block associated with one axis.
Substituting Eqs.~\eqref{eq:A} and \eqref{eq:B} into Eq.~\eqref{eq:H-aba}, we obtain
\begin{align}
	\mathcal{H}
	= \sum_{\alpha=1}^{3}\sum_{k=1}^{n}B_{\mathrm{cell}}^{-\frac{1}{2}}A^{(\alpha)} B_{\mathrm{cell}}^{-\frac{1}{2}}\otimes \frac{i}{2h}S_k^{(\alpha)} \label{eq:H-axis-decomposition}.
\end{align}
For each $\alpha \in \{1,2,3\}$, $B_{\mathrm{cell}}^{-\frac{1}{2}}A^{(\alpha)} B_{\mathrm{cell}}^{-\frac{1}{2}}$ is Hermitian on the 4-qubit state register, and thus admits the eigendecomposition
\begin{align}
	B_{\text{cell}}^{-\frac{1}{2}}A^{(\alpha)} B_{\text{cell}}^{-\frac{1}{2}}
	= \sum_{j=0}^{15} \lambda_j^{(\alpha)}\ketbra*{\phi_j^{(\alpha)}}
	\label{eq:eigendecomposition}
\end{align}
where $\{\ket*{\phi_j^{(\alpha)}}\}_{j=0}^{15}$ is an orthonormal set of eigenstates, and $\lambda_j^{(\alpha)}\in \mathbb{R}$ are the corresponding eigenvalue.
Defining $P_j^{(\alpha)}:=\ketbra*{\phi_j^{(\alpha)}}$, we rewrite $\mathcal{H}$ as
\begin{align}
\mathcal{H}
&= \sum_{\alpha=1}^{3}\sum_{j=0}^{15}\sum_{k=1}^{n}
H_{jk}^{(\alpha)},\quad H_{jk}^{(\alpha)}
:= P_j^{(\alpha)}\otimes \frac{i\lambda_j^{(\alpha)}}{2h}S_k^{(\alpha)}.
\label{eq:H-decomposition-with-48n-components}
\end{align}

The one step first-order Trotter operator is defined by
\begin{align}
    U_1(\tau)
    :=\prod_{\alpha =1}^{3}\prod_{j=0}^{15}\prod_{k=1}^{n} e^{-iH_{jk}^{(\alpha)}\tau}\label{eq:U1}
\end{align}
where the product order is fixed as
\begin{align}
	\prod_{\alpha=1}^{3}\prod_{j=0}^{15}\prod_{k=1}^{n}F_{jk}^{(\alpha)}
	:=F_{15,n}^{(3)}F_{15,n-1}^{(3)}\cdots F_{0,1}^{(1)},
\end{align}
for an operator family $\{ F_{jk}^{(\alpha)} \}_{j, k, \alpha}$.
Applying the first-order Trotter formula~\cite{UniversalQuantumSimulators1996}, for sufficiently small $\tau$ we have
\begin{align}
\left\|e^{-i\mathcal{H}\tau}-U_1(\tau)\right\| = O(\tau^2),
\end{align}
where $\norm*{\cdot}$ denotes the operator norm.

Since $P_j^{(\alpha)}$ is a projector $(P_j^{(\alpha)})^2=P_j^{(\alpha)}$, each single term exponential is rewritten as
\begin{align}\label{eq:hermite-term-exponential-map-controlled}
    &\exp\left(-iH_{jk}^{(\alpha)}\tau\right)\notag\\
    & =\exp\left(-iP_j^{(\alpha)}\otimes \frac{\lambda_j^{(\alpha)}\tau}{2h} S_k^{(\alpha)}\right)\notag\\
    & =P_j^{(\alpha)}\otimes W_{jk}^{(\alpha)}(\tau) + \left(I^{\otimes 4}-P_j^{(\alpha)}\right)\otimes I^{\otimes 3n},
\end{align}
where an operator $W_{jk}^{(\alpha)}(\tau)$ on the spatial register is denoted by
\begin{align}
    W_{jk}^{(\alpha)}(\tau)=\exp\!\left(-i\theta_{jk}^{(\alpha)}(\tau)S_k^{(\alpha)}\right),\quad
\theta_{jk}^{(\alpha)}(\tau):=\frac{\lambda_j^{(\alpha)}\tau}{2h}.
\end{align}
This gives a projector controlled implementation form for each term.
Following Ref.~\cite{Hamiltoniansimulationhyperbolicpartialdifferentialequationsscalablequantumcircuits2024}, the operator $W_{jk}^{(\alpha)}(\tau)$ is written as
\begin{align}
	&W_{jk}^{(\alpha)}(\tau) =  I^{\otimes (n(\alpha-1)+n-k)}\otimes\\
    &\qquad\qquad\left( G_k \crz_{k}^{1,\ldots,k-1}\left(-\theta_{jk}^{(\alpha)}(\tau)\right) G_{k}^{\dagger}\right) \otimes I^{\otimes n(3-\alpha)},\label{eq:W-gamma-alpha}
\end{align}
where we define
\begin{align}
	&\crz_{k}^{1,\ldots,k-1}(\theta)\notag\\ 
    &:= \exp(iZ\theta)\otimes \ketbra*{1}^{\otimes (k-1)} + I \otimes (I^{\otimes (k-1)} - \ketbra*{1}^{\otimes (k-1)})
\end{align}
Here, the superscript \(1,\ldots,k-1\) indicates the control qubits, while the subscript \(k\) denotes the target qubit on which the \(R_z\) rotation acts.
and
\begin{align}
	G_k := \left(\prod_{l=1}^{k-1}\cnot_{l}^{k}\right)\mathrm{S}_k \mathrm{H}_k.
\end{align}
Here, we denote the Hadamard gate and S gate as
\begin{align}
	\mathrm{H}_k =\frac{1}{\sqrt{2}} 
	\begin{bmatrix}
		1&1\\
		1&-1
	\end{bmatrix}
	\otimes I^{\otimes k-1}
	,\quad
	\mathrm{S}_k = 
	\begin{bmatrix}
		1&0\\
		0&i
	\end{bmatrix}
	\otimes I^{\otimes k-1}.
\end{align}

Since direct control by the projector $P_j^{(\alpha)}$ is not directly implementable in the computational basis, we convert it into computational basis form by a basis change on the 4-qubit state register.
Let $\left\{\ket*{j}\right\}_{j=0}^{15}$ be the computational basis, and define
\begin{align}\label{eq:changeunitary}
	V^{(\alpha)} := \sum_{j=0}^{15}\ketbra*{\phi_j^{(\alpha)}}{j}.
\end{align}
Then, $V^{(\alpha)}\ket*{j}=\ket*{\phi_j^{(\alpha)}}$, and hence
\begin{align}
    P_j^{(\alpha)} = V^{(\alpha)}\ketbra*{j}V^{(\alpha)\dagger}.
\end{align}
Using this relation, define
\begin{widetext}
\begin{align}
\mathcal{W}_{j}^{(\alpha)}(\tau)
:=
\left(V^{(\alpha)} \otimes I^{\otimes 3n}\right)
\left(\prod_{k=1}^{n}
\left(
\ketbra*{j}\otimes W_{jk}^{(\alpha)}(\tau)
+\left(I^{\otimes 4}-\ketbra*{j}\right)\otimes I^{\otimes 3n}
\right)\right)
\left(V^{(\alpha)\dagger} \otimes I^{\otimes 3n}\right).\label{eq:basis-change-by-V}
\end{align}
\end{widetext}
Then the first-order operator is written as
\begin{align}
	U_1(\tau)=
	\prod_{\alpha =1}^{3}\prod_{j=0}^{15}\mathcal{W}_{j}^{(\alpha)}(\tau).
	\label{eq:U1-circuit}
\end{align}
Here, $W_{jk}^{(\alpha)}(\tau)$ is available in explicit circuit form in Ref.~\cite{Hamiltoniansimulationhyperbolicpartialdifferentialequationsscalablequantumcircuits2024}, and $V^{(\alpha)}$ acts on the 4-qubit state register, so its synthesis into single gate and CNOT gates is standard.
Therefore, Eq.~\eqref{eq:basis-change-by-V} gives an explicit gate-level implementation of $U_1(\tau)$, and the corresponding circuit structure is shown in Fig.~\ref{fig:single-circuit-j1-j2-j15}.
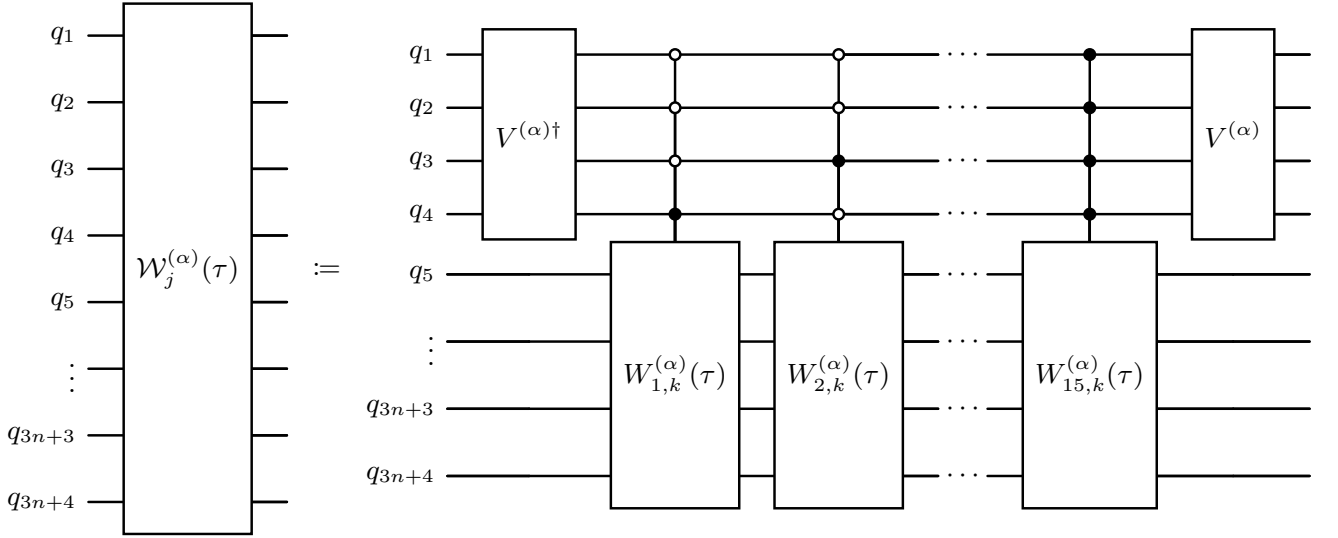
\begin{figure*}[t]
\centering
\resizebox{0.98\textwidth}{!}{$%
  \vcenter{\hbox{%
    \begin{quantikz}[row sep=0.03cm, column sep=0.40cm]
      \lstick{$q_1$}      & \gate[wires=8]{\mathcal{W}_{j}^{(\alpha)}(\tau)} & \qw \\
      \lstick{$q_2$}      & \ghost{\mathcal{W}_{j}^{(\alpha)}(\tau)}          & \qw \\
      \lstick{$q_3$}      & \ghost{\mathcal{W}_{j}^{(\alpha)}(\tau)}          & \qw \\
      \lstick{$q_4$}      & \ghost{\mathcal{W}_{j}^{(\alpha)}(\tau)}          & \qw \\
      \lstick{$q_5$}      & \ghost{\mathcal{W}_{j}^{(\alpha)}(\tau)}          & \qw \\
      \lstick{$\vdots$}   & \ghost{\mathcal{W}_{j}^{(\alpha)}(\tau)}          & \qw \\
      \lstick{$q_{3n+3}$} & \ghost{\mathcal{W}_{j}^{(\alpha)}(\tau)}          & \qw \\
      \lstick{$q_{3n+4}$} & \ghost{\mathcal{W}_{j}^{(\alpha)}(\tau)}          & \qw
    \end{quantikz}%
  }}%
  \;\mathrel{\coloneqq}\;
  \vcenter{\hbox{%
    \begin{quantikz}[row sep=0.03cm, column sep=0.40cm]
      \lstick{$q_1$} & \gate[wires=4]{V^{(\alpha)\dagger}}
        & \octrl{4} & \octrl{4} & \qw \push{\;\cdots\;}
        & \ctrl{4}
        & \gate[wires=4]{V^{(\alpha)}} & \qw \\
      \lstick{$q_2$} & \ghost{V^{(\alpha)\dagger}}
        & \octrl{3} & \octrl{3} & \qw \push{\;\cdots\;}
        & \ctrl{3}
        & \ghost{V^{(\alpha)}} & \qw \\
      \lstick{$q_3$} & \ghost{V^{(\alpha)\dagger}}
        & \octrl{2} & \ctrl{2}  & \qw \push{\;\cdots\;}
        & \ctrl{2}
        & \ghost{V^{(\alpha)}} & \qw \\
      \lstick{$q_4$} & \ghost{V^{(\alpha)\dagger}}
        & \ctrl{1}  & \octrl{1} & \qw \push{\;\cdots\;}
        & \ctrl{1}
        & \ghost{V^{(\alpha)}} & \qw \\
      \lstick{$q_5$} & \qw
        & \gate[wires=4]{W^{(\alpha)}_{1,k}(\tau)}
        & \gate[wires=4]{W^{(\alpha)}_{2,k}(\tau)}
        & \qw \push{\;\cdots\;}
        & \gate[wires=4]{W^{(\alpha)}_{15,k}(\tau)}
        & \qw & \qw \\
      \lstick{$\vdots$} & \qw
        & \ghost{W^{(\alpha)}_{1,k}(\tau)}
        & \ghost{W^{(\alpha)}_{2,k}(\tau)}
        & \qw \push{\;\cdots\;}
        & \ghost{W^{(\alpha)}_{15,k}(\tau)}
        & \qw & \qw \\
      \lstick{$q_{3n+3}$} & \qw
        & \ghost{W^{(\alpha)}_{1,k}(\tau)}
        & \ghost{W^{(\alpha)}_{2,k}(\tau)}
        & \qw \push{\;\cdots\;}
        & \ghost{W^{(\alpha)}_{15,k}(\tau)}
        & \qw & \qw \\
      \lstick{$q_{3n+4}$} & \qw
        & \ghost{W^{(\alpha)}_{1,k}(\tau)}
        & \ghost{W^{(\alpha)}_{2,k}(\tau)}
        & \qw \push{\;\cdots\;}
        & \ghost{W^{(\alpha)}_{15,k}(\tau)}
        & \qw & \qw
    \end{quantikz}%
  }}%
$}%
\caption{Circuit structure of the block \(\mathcal{W}_{j}^{(\alpha)}(\tau)\) in Eq.~\eqref{eq:basis-change-by-V}. The basis change operators \(V^{(\alpha)}\) and \(V^{(\alpha)\dagger}\) act on the 4-qubit state register, and each controlled block \(W_{jk}^{(\alpha)}(\tau)\) acts on the spatial register.}
\label{fig:single-circuit-j1-j2-j15}
\end{figure*}

Using the same construction as in the first-order case, we also define the second-order (symmetric) Trotter operator by
\begin{align}
    U_2(\tau)
    :=\prod_{\alpha =1}^{3}\prod_{j=0}^{15}\prod_{k=1}^{n} e^{-iH_{jk}^{(\alpha)}\tau/2}\prod_{\alpha =3}^{1}\prod_{j=15}^{0}\prod_{k=n}^{1} e^{-iH_{jk}^{(\alpha)}\tau/2}.\label{eq:U2}
\end{align}
In terms of the circuit blocks introduced above, this is equivalently written as
\begin{align}
	U_2(\tau)=
	\prod_{\alpha=1}^{3}\prod_{j=0}^{15}\mathcal{W}_{j}^{(\alpha)}(\tau/2)\prod_{\alpha'=3}^{1}\prod_{j'=15}^{0}\mathcal{W}_{j'}^{(\alpha')\,\mathrm{rev}}(\tau/2).\label{eq:U2-implement}
\end{align}
Here, $\mathcal{W}_{j'}^{(\alpha')\,\mathrm{rev}}(\tau/2)$ denotes the same block as $\mathcal{W}_{j}^{(\alpha)}(\tau/2)$ with the product order reversed from $k=1,\ldots,n$ to $k=n,\ldots,1$.
In the small-$\tau$ regime, the approximation error satisfies
\begin{align}
	\left\|e^{-i\mathcal{H}\tau}-U_2(\tau)\right\|=O(\tau^3),
\end{align}
where $\norm*{\cdot}$ is the operator norm~\cite{EfficientquantumalgorithmssimulatingsparseHamiltonians2007}.
\subsection{Complexity and Error Analysis}
In this subsection, we derive one step Trotter error bounds and convert them into global CNOT complexity bounds.

We first analyze one step approximation errors as functions of the step size $\tau$.
Then, for a total simulation time $T$ and target accuracy $\epsilon$, we choose the number of Trotter steps $m$ and set $\tau = T/m$.
\subsubsection{The first order Trotter formula}
We analyze the first-order formula: one step error bound, one step implementation cost, and the resulting global complexity.
We start from the standard commutator scaling bound~\cite{TheoryTrotterErrorCommutatorScaling2021a}.
\begin{proposition}[Commutator scaling bound for the first-order Trotter formula: Proposition 9 of Ref.~\cite{TheoryTrotterErrorCommutatorScaling2021a}]\label{prop:frist-order-trotter-error-original}
Let
\begin{align}
	\mathcal{H}=\sum_{r=1}^{R}H_r,
\end{align}
where each $H_r$ is Hermitian.
Define the first-order product formula by
\begin{align}
	U_1(\tau):=\prod_{r=1}^{R}e^{-iH_{r}\tau}.
\end{align}
Then the first-order Trotter error satisfies
\begin{align}
    \norm*{U_1(\tau)- e^{-i\mathcal{H}\tau}} \leq \frac{\tau^2}{2}\sum_{r_1=1}^{R}\norm\Bigg{\left[\sum_{r_2=r_1+1}^{R}H_{r_2},H_{r_1}\right]}.\label{eq:commutator-scaling}
\end{align}
In particular, by triangle inequality and $\norm*{[A,B]}\leq 2\norm*{A}\norm*{B}$, we obtain
\begin{align}
    \norm*{U_1(\tau)- e^{-i\mathcal{H}\tau}}
    &\leq \frac{\tau^2}{2}\left(\sum_{r=1}^{R}\norm*{H_{r}}\right)^2.\label{eq:norm-scaling}
\end{align}
\end{proposition}
\begin{proof}
The detailed proof of Eq.~\eqref{eq:commutator-scaling} is given in Ref.~\cite{TheoryTrotterErrorCommutatorScaling2021a}.
Applying the triangle inequality to Eq.~\eqref{eq:commutator-scaling} gives 
\begin{align}
    \norm\Bigg{U_1(\tau)- e^{-i\mathcal{H}\tau}} 
    &\leq \frac{\tau^2}{2}\sum_{r_1=1}^{R}\norm\Bigg{\left[\sum_{r_2=r_1+1}^{R}H_{r_2},H_{r_1}\right]}\notag\\
    &\leq \frac{\tau^2}{2}\sum_{r_1=1}^{R}\sum_{r_2=r_1+1}^{R}\norm*{\left[H_{r_2},H_{r_1}\right]}\notag\\
    &\leq \frac{\tau^2}{2}\sum_{r_1=1}^{R}\sum_{r_2=r_1+1}^{R}2\norm*{H_{r_2}}\norm*{H_{r_1}}\notag\\
    &\leq \frac{\tau^2}{2}\left(\sum_{r=1}^{R}\norm*{H_{r}}\right)^2.
\end{align} 
\end{proof}
We first derive a baseline bound from Eq.~\eqref{eq:norm-scaling}; a sharper error bound based on Eq.~\eqref{eq:commutator-scaling} is given later.
\begin{proposition}[First-order Trotter error bound via norm scaling]\label{thm:1st-trotter-error-norm-scaling}
	Let $\mathcal{H}$ be given by Eq.~\eqref{eq:H-decomposition-with-48n-components}, and let $U_1(\tau)$ be defined in Eq.~\eqref{eq:U1}.
	Then, by Eq.~\eqref{eq:norm-scaling}, the one step approximation is bounded as
    \begin{align}\label{eq:1st-trotter-error-norm-scaling}
		\norm*{U_1(\tau)- e^{-i\mathcal{H}\tau}}
		\leq
		\frac{3^4\tau^2}{2h^2}\rho^{-1}\norm*{\mathcal{S}^{-1}_{\mathrm{comp}}}n^2.
	\end{align}
\end{proposition}
\begin{proof}
By Eq.~\eqref{eq:norm-scaling},
\begin{align}
\left\|e^{-i\mathcal{H}\tau}-U_1(\tau)\right\|
\le
\frac{\tau^2}{2}
\left(
\sum_{\alpha,j,k}
\left\|H_{jk}^{(\alpha)}\right\|
\right)^2.
\label{eq:p-main}
\end{align}
For each $(\alpha,j,k)$,
\begin{align}
	\norm*{H_{jk}^{(\alpha)}}
	&=
		\norm*{P_j^{(\alpha)}\otimes \frac{i\lambda_j^{(\alpha)}}{2h}S_k^{(\alpha)}}\notag\\
	&=
		\frac{\abs*{\lambda_j^{(\alpha)}}}{2h}\norm*{P_j^{(\alpha)}}\norm*{S_k^{(\alpha)}}\notag\\
	&=
		\frac{\abs*{\lambda_j^{(\alpha)}}}{2h}
	\label{eq:p-local}
\end{align}
where we used $\norm*{A\otimes B}=\norm*{A}\norm*{B}$, $\norm*{P_j^{(\alpha)}}=1$, and 
\begin{align}
	\norm*{S_k^{(\alpha)}}=1
\end{align}
from Eq.~(34) of Ref.~\cite{Hamiltoniansimulationhyperbolicpartialdifferentialequationsscalablequantumcircuits2024}.
Therefore,
\begin{align}
\sum_{\alpha=1}^{3}\sum_{j=0}^{15}\sum_{k=1}^{n}\left\|H_{jk}^{(\alpha)}\right\|
&=
\frac{1}{2h}
\sum_{\alpha=1}^{3}\sum_{j=0}^{15}\sum_{k=1}^{n}
|\lambda_j^{(\alpha)}| \notag\\
&=
\frac{n}{2h}
\sum_{\alpha=1}^{3}\sum_{j=0}^{15}
|\lambda_j^{(\alpha)}| \notag\\
&=
\frac{n}{2h}
\sum_{\alpha=1}^{3}
\trnorm{B_{\mathrm{cell}}^{-1/2}A^{(\alpha)}B_{\mathrm{cell}}^{-1/2}}.
\label{eq:p-sum}
\end{align}
Applying Lemma~\ref{lemma:BAB-bounded} to Eq.~\eqref{eq:p-sum}, and substituting into Eq.~\eqref{eq:p-main}, we obtain
\begin{align}
\left\|U_1(\tau)-e^{-i\mathcal{H}\tau}\right\|
\le
\frac{3^4\tau^2}{2h^2}\rho^{-1}\left\|\mathcal{S}_{\mathrm{comp}}^{-1}\right\|n^2,
\end{align}
which is Eq.~\eqref{eq:1st-trotter-error-norm-scaling}.
\end{proof}
To derive the gate complexity, we next estimate the one step implementation cost of $U_1(\tau)$.
This cost bound is independent of the norm scaling error estimate above.
\begin{lemma}[One step implementation cost for $U_1$]\label{lemma:one-step-U1}
	Let $U_1(\tau)$ be the first-order Trotter operator defined in Eq.~\eqref{eq:U1}.
	Assume $U_1(\tau)$ is implemented by the circuit in Eq.~\eqref{eq:basis-change-by-V}.
	Then one application of $U_1(\tau)$ uses
	\begin{align}
		3n+4
	\end{align}
	qubits and at most
	\begin{align}
		432n^2+378
	\end{align}
	CNOT gates.
\end{lemma}
\begin{proof}
Equation~\eqref{eq:basis-change-by-V} gives an explicit one step implementation, so the qubit count $N_{\mathrm{qubit}}$ is
\begin{align}
	N_{\mathrm{qubit}} = 3n+4.
\end{align}
For CNOT counting, Eq.~\eqref{eq:hermite-term-exponential-map-controlled} applies state register control to the whole $W_{jk}^{(\alpha)}(\tau)$ block, and Eq.~\eqref{eq:W-gamma-alpha} shows that, by symmetry, the effective controlled part is only $\crz_k^{1,\dots,k-1}$.
Hence, the rotation gate to be implemented has $(k+3)$ control qubits.
Write
\begin{align}
	N_{\mathrm{ladder}}(k) = 2(k-1),\quad
	N_{\mathrm{CRZ}}(k) = 16(k+3) -40+3,
\end{align}
where $N_{\mathrm{ladder}}(k)$ is the CNOT ladder cost from $G_k, G_k^{\dagger}$, and $N_{\mathrm{CRZ}}(k)$ is the controlled $R_z$ cost for $\crz_k^{1,\dots,k-1}$.
$\crz_k^{1,\dots,k-1}$ is synthesized with $16(k+3) -40$ CNOT gates and one controlled single qubit rotation gate~\cite{CircuitDecompositionMulticontrolledSpecialUnitarySingleQubitGates2024}, and the latter is implemented with $3$ CNOT gates~\cite{QuantumComputationQuantumInformation10thAnniversaryEdition2010a}. 
For basis change on the state register, each of $V^{(\alpha)}$ and $V^{(\alpha)\dagger}$ costs at most
\begin{align}
	N_V \le 63
\end{align}
CNOT gates~\cite{TheoryTrotterErrorCommutatorScaling2021a}.
Within each fixed $\alpha$, adjacent $V^{(\alpha)\dagger}V^{(\alpha)}$ pairs cancel across $j$, leaving one pair per $\alpha$.
Therefore, the total number of CNOT gates $N_{\mathrm{CNOT}}$ is
	\begin{align}
		N_{\mathrm{CNOT}} 
		&= \sum_{\alpha=1}^{3}\left[\sum_{j=0}^{15}\sum_{k=1}^{n}\left(N_{\mathrm{ladder}}(k)+N_{\mathrm{CRZ}}(k)\right) + 2N_V\right]\notag\\
		&\le 432n^2+378.
	\end{align}
	Here one application of $U_1(\tau)$ uses $3n+4$ qubits and at most $432n^2+378$ CNOT gates.
\end{proof}
To pass from one step bounds to global gate complexity, we use the following generic conversion lemma.
\begin{lemma}\label{lemma:error-to-gate-conversion}
Let $\mathcal{H}$ be a Hamiltonian, and let $U(\tau)$ be a one step approximation of $e^{-i\mathcal{H}\tau}$.
Assume that for some constants $C>0$ and $p>1$,
\begin{align}
	\norm*{e^{-i\mathcal{H}\tau}-U(\tau)} \le C\,\tau^p .
\end{align}
Assume also that one application of $U(\tau)$ uses at most $N_{\mathrm{step}}$ CNOT gates.

Then, for simulation time $T>0$ and target error $\epsilon>0$, if
\begin{align}
	m \in \mathbb{N},\qquad
	m\ge \left(\frac{CT^p}{\epsilon}\right)^{\frac{1}{p-1}},
\end{align}
the m-step product $U(T/m)^m$ satisfies
\begin{align}
	\norm*{e^{-i\mathcal{H}T}-U(T/m)^m} \le \epsilon,
\end{align}
and the required CNOT gate count $N_{\mathrm{CNOT}}$ satisfies
\begin{align}
N_{\mathrm{CNOT}}
&\le
\left\lceil
\left(\frac{CT^p}{\epsilon}\right)^{\frac{1}{p-1}}
\right\rceil N_{\mathrm{step}}\notag\\
&\backsimeq 
\left(\frac{CT^p}{\epsilon}\right)^{\frac{1}{p-1}}N_{\mathrm{step}}.
\end{align}

\end{lemma}
\begin{proof}
Setting $\tau=T/m$,
\begin{align}
\norm*{e^{-i\mathcal{H}T}-U(\tau)^m}
&\le
m\,\norm*{e^{-i\mathcal{H}\tau}-U(\tau)}\notag\\
&\le
m\,C\left(\frac{T}{m}\right)^p\notag\\
&=
\frac{CT^p}{m^{p-1}}.
\end{align}
Hence $m\ge (CT^p/\epsilon)^{1/(p-1)}$ implies error within $\epsilon$.
Taking the smallest such integer $m$ gives
\begin{align}
	m=\left\lceil\left(\frac{CT^p}{\epsilon}\right)^{\frac{1}{p-1}}\right\rceil,
\end{align}
and therefore
\begin{align}
	N_{\mathrm{CNOT}}
	\le
		mN_{\mathrm{step}}
		\le
		\left\lceil\left(\frac{CT^p}{\epsilon}\right)^{\frac{1}{p-1}}\right\rceil N_{\mathrm{step}}.
\end{align}

\end{proof}
Applying Lemma~\ref{lemma:error-to-gate-conversion} with one step error bound from Proposition~\ref{thm:1st-trotter-error-norm-scaling} (i.e., $p=2$) and the one step CNOT cost from Lemma~\ref{lemma:one-step-U1}, we obtain the following global CNOT complexity bound.
\begin{theorem}[First-order Trotter gate complexity: norm scaling baseline]\label{thm:1st-trotter-gates-norm-scaling}
	Consider the Hamiltonian $\mathcal{H}$ in Eq.~\eqref{eq:H-decomposition-with-48n-components}.
	Using the first-order Trotter formula, the time evolution $e^{-i\mathcal{H}T}$ can be approximated within additive error $\epsilon$ (in operator norm) by a circuit using
	\begin{align}
		3n+4
	\end{align}
	qubits and at most
    \begin{align}\label{eq:1st-trotter-gates}
		\frac{3^4T^2}{h^2\epsilon}\rho^{-1}\norm*{\mathcal{S}^{-1}_{\mathrm{comp}}}n^2\left(216n^2+189\right)
	\end{align}
	CNOT gates.
		Here, $2^n$ is the number of grid points per coordinate, $h$ is the grid interval, $\rho$ is the mass density, and $S_{\mathrm{comp}}$ is the compliance matrix.
\end{theorem}
\begin{proof}
By Proposition~\ref{thm:1st-trotter-error-norm-scaling},
\begin{align}
	\norm*{e^{-i\mathcal{H}\tau}-U_1(\tau)}
	\le
	C\tau^2, \qquad
	C=
		\frac{3^4}{2h^2}\rho^{-1}\norm*{\mathcal{S}_{\mathrm{comp}}^{-1}}\,n^2.
\end{align}
By Lemma~\ref{lemma:one-step-U1}, one step uses $3n+4$ qubits and at most \(432n^2+378\) CNOT gates.
Applying Lemma~\ref{lemma:error-to-gate-conversion} with \(N_{\mathrm{step}}=432n^2+378\), we obtain Eq.~\eqref{eq:1st-trotter-gates}.

\end{proof}
We next derive a sharper bound from Eq.~\eqref{eq:commutator-scaling}.
\begin{proposition}[First-order Trotter error bound via commutator scaling]\label{thm:1st-trotter-error}
	Let $\mathcal{H}$ be given by Eq.~\eqref{eq:H-decomposition-with-48n-components}, and let $U_1(\tau)$ be defined in Eq.~\eqref{eq:U1}.
Then, by Eq.~\eqref{eq:commutator-scaling}, the one step approximation is bounded as
    \begin{align}\label{eq:1st-trotter-error}
		\norm*{U_1(\tau)- e^{-i\mathcal{H}\tau}}
		\leq
			\frac{9\tau^2}{4h^2}\rho^{-1}\norm*{\mathcal{S}^{-1}_{\mathrm{comp}}}\left(5n-1\right).
	\end{align}
\end{proposition}
\begin{proof}[Proof sketch]
From Proposition~\ref{prop:frist-order-trotter-error-original},
\begin{align}
\norm*{U_1(\tau)- e^{-i\mathcal{H}\tau}}
\le
\frac{\tau^2}{2}
\left[
\sum_{\alpha=1}^{3}\Xi_{\alpha}^{\mathrm{same}}
+
\sum_{\alpha<\beta}\Xi_{\alpha\beta}^{\mathrm{cross}}
\right],
\label{eq:t-error-in-main-paper}
\end{align}
where
\begin{align}
\Xi_{\alpha}^{\mathrm{same}}
&:=
\sum_{j=0}^{15}\sum_{k=1}^{n}
\norm\Bigg{\left[H^{(\alpha)}_{jk},
\sum_{l=1}^{15}\sum_{m=1}^{n}H^{(\alpha)}_{lm}
+\sum_{m=k+1}^{n}H^{(\alpha)}_{jm}\right]},\\
\Xi_{\alpha\beta}^{\mathrm{cross}}
&:=
\sum_{j=0}^{15}\sum_{k=1}^{n}
\norm\Bigg{\left[H^{(\alpha)}_{jk},
\sum_{l=0}^{15}\sum_{m=1}^{n}H^{(\beta)}_{lm}\right]}.
\end{align}
For each $\alpha$, the same axis part is bounded by
\begin{align}
\Xi_{\alpha}^{\mathrm{same}}
\le
\frac{n-1}{4h^2}
\trnormBigg{
\left(B_{\mathrm{cell}}^{-1/2}A^{(\alpha)}B_{\mathrm{cell}}^{-1/2}\right)^2
}.
\label{eq:normsum1-main}
\end{align}
For each $\alpha>\beta$, the cross axis part is bounded by
\begin{align}
\Xi_{\alpha\beta}^{\mathrm{cross}}
\le
\frac{n}{h^2}
\trnormBigg{B_{\mathrm{cell}}^{-1/2}A^{(\alpha)}B_{\mathrm{cell}}^{-1/2}}
\norm*{B_{\mathrm{cell}}^{-1/2}A^{(\beta)}B_{\mathrm{cell}}^{-1/2}}.
\label{eq:normsum2-main}
\end{align}
Applying Lemma~\ref{lemma:BAB-bounded} to Eqs.~\eqref{eq:normsum1-main} and \eqref{eq:normsum2-main}, we obtain
\begin{align}
\Xi_{\alpha}^{\mathrm{same}}
&\le
\frac{n-1}{4h^2}\cdot 6\rho^{-1}\norm*{\mathcal{S}_{\mathrm{comp}}^{-1}},\\
\Xi_{\alpha\beta}^{\mathrm{cross}}
&\le
\frac{n}{h^2}\cdot 6\rho^{-1}\norm*{\mathcal{S}_{\mathrm{comp}}^{-1}}.
\end{align}
Substituting these bounds into Eq.~\eqref{eq:t-error-in-main-paper} yields Eq.~\eqref{eq:1st-trotter-error}.
(Full derivation is given in Appendix~\ref{app:first-order-commutator}.)
\end{proof}
Combining Proposition~\ref{thm:1st-trotter-error}, Lemma~\ref{lemma:one-step-U1}, and Lemma~\ref{lemma:error-to-gate-conversion}, we obtain global CNOT complexity bound.
\begin{theorem}[First-order Trotter gate complexity: commutator scaling]\label{thm:1st-trotter-gates}
    Consider $\mathcal{H}$ in Eq.~\eqref{eq:H-decomposition-with-48n-components}.
    Using the first-order Trotter formula, the time evolution operator $e^{-i\mathcal{H}T}$ can be approximated within additive error $\epsilon$ (in operator norm) by a circuit using
	\begin{align}
		3n+4
	\end{align}
	qubits and at most
    \begin{align}\label{eq:1st-trotter-gates0}
			\frac{9T^2}{2h^2\epsilon}\rho^{-1}\norm*{\mathcal{S}^{-1}_{\mathrm{comp}}}\left(1080n^3-216n^2+945n-189\right)
	\end{align}
	CNOT gates.
		Here, $2^n$ is the number of grid points per coordinate, $h$ is the grid interval, $\rho$ is the mass density, and $S_{\mathrm{comp}}$ is the compliance matrix.
\end{theorem}
\begin{proof}
By Proposition~\ref{thm:1st-trotter-error},
\begin{align}
	\norm*{e^{-i\mathcal{H}\tau}-U_1(\tau)}\le C\tau^2, \\
	C = \frac{9}{4h^2}\rho^{-1}\norm*{\mathcal{S}_{\mathrm{comp}}^{-1}}(5n-1).
\end{align}
By Lemma~\ref{lemma:one-step-U1}, one step uses $3n+4$ qubits and at most \(432n^2+378\) CNOT gates.
Applying Lemma~\ref{lemma:error-to-gate-conversion} with \(N_{\mathrm{step}}=432n^2+378\), we obtain Eq.~\eqref{eq:1st-trotter-gates0}.
\end{proof}
Comparing Eq.~\eqref{eq:1st-trotter-gates} and Eq.~\eqref{eq:1st-trotter-gates0}, for fixed $T,\epsilon, h,\rho,S_{\mathrm{comp}}$, the leading $n$-dependence of the CNOT upper bound improves from $O(n^4)$ (norm scaling) to $O(n^3)$ (commutator scaling).
\begin{remark}[Material parameter interpretation]
To interpret the bounds in material parameters, we assume the linear isotropic medium.
The compliance matrix is
\begin{align}
    S_{\mathrm{comp}}=\frac{1}{E}
    \begin{pmatrix}
        1 & -\nu & -\nu &0&0&0\\
        -\nu & 1 & -\nu &0&0&0\\
        -\nu & -\nu & 1 &0&0&0\\
        0&0&0& 1+\nu &0&0\\
        0&0&0& 0 &1+\nu&0\\
        0&0&0& 0 &0&1+\nu\\
    \end{pmatrix}.
	\label{eq:Scomp-of-isotropic}
\end{align}
with the Young's modulus $E>0$ and the Poisson ratio $-1<\nu <\frac{1}{2}$.
Hence
\begin{align}
    \norm*{S_{\mathrm{comp}}^{-1}}
    =
    \begin{cases}
        \frac{E}{1-2\nu} & \text{ for } 0\leq \nu <\frac{1}{2},\\
        \frac{E}{1+\nu}  & \text{ for } -1<\nu < 0.
    \end{cases}
\end{align}
The details are given in Appendix~\ref{app:isotropic-compliance}.
Equivalently, we can rewrite
\begin{align}
	\norm*{S_{\mathrm{comp}}^{-1}}
	=
	\max\{3K,2\mu\},
\end{align}
where the bulk modulus $K$ and shear modulus $\mu$ are given by
\begin{align}
	K = \frac{E}{3(1-2\nu)},\quad
	\mu = \frac{E}{2(1+\nu)}.
\end{align}
Therefore, the error and gate bound increase with material stiffness (large $K,\mu$) and decrease with mass density $\rho$ through factors such as $\rho^{-1}\norm*{S_{\mathrm{comp}}^{-1}}$.
In particular, near the compressible limit $\nu\to\frac12$, the factor $E/(1-2\nu)$ becomes large, which worsens these bounds.
\end{remark}
\subsubsection{The second order Trotter formula}
We next analyze the second-order formula $U_2(\tau)$.
As in the first-order case, we first state the generic norm scaling error bound, and then specialize it to the Hamiltonian in Eq.~\eqref{eq:H-decomposition-with-48n-components} to obtain explicit parameter dependence.
\begin{proposition}[Norm scaling bound for the second-order Trotter formula: Lemma 1 in Ref.~\cite{EfficientquantumalgorithmssimulatingsparseHamiltonians2007}]\label{prop:second-order-trotter-error}
	Let
	\begin{align}
		\mathcal{H} = \sum_{r=1}^{R}H_r,
	\end{align}
	where each $H_r$ is Hermitian. Define the generic second-order Trotter formula by
	\begin{align}
		U_2(\tau):= \prod_{r=1}^{R}e^{-iH_r \tau/2}\prod_{r=R}^{1}e^{-iH_r \tau/2}.
	\end{align}
	Then the second-order Trotter error satisfies
    \begin{align}
        \norm*{e^{-i\mathcal{H}\tau} - U_2(\tau)} \leq
        2(10R \norm*{\mathcal{H}} \tau)^{3},\label{eq:general-norm-scaling-bound}
    \end{align}
	provided that	
    \begin{align}
        (10R \norm*{\mathcal{H}}\tau)^3 \leq 1.
    \end{align}
\end{proposition}
We now specialize Proposition~\ref{prop:second-order-trotter-error} to the Hamiltonian in Eq.~\eqref{eq:H-decomposition-with-48n-components}.
Substituting $R=48n$ and $\norm*{\mathcal{H}}$ yields the following explicit one step error bound.
\begin{proposition}[Second-order Trotter error bound via norm scaling]\label{thm:2nd-trotter-error}
	Let $\mathcal{H}$ be given in Eq.\eqref{eq:H-decomposition-with-48n-components}, and let $U_2(\tau)$ be defined in Eq.~\eqref{eq:U2}.
Then, by Eq.~\eqref{eq:general-norm-scaling-bound}, the one step approximation error is bounded as
    \begin{align}\label{eq:2nd-trotter-error}
		\norm*{U_2(\tau)- e^{-i\mathcal{H}\tau}}
		\leq
		\frac{2\cdot 1440^{3} \tau^{3}}{h^3} \rho^{-\frac{3}{2}} \norm*{S_{\mathrm{comp}}^{-1}}^{\frac{3}{2}}n^3.
	\end{align} 
\end{proposition}
\begin{proof}
To apply Proposition~\ref{prop:second-order-trotter-error}, we first bound $\mathcal{H}$.
From Eq.~\eqref{eq:H-axis-decomposition}, write
\begin{align}
\mathcal{H}
=
\sum_{\alpha=1}^{3}
\left(B_{\mathrm{cell}}^{-\frac12}A^{(\alpha)}B_{\mathrm{cell}}^{-\frac12}\right)\otimes D^{(\alpha)}.
\end{align}
Here, by submultiplicativity of the operator norm,
\begin{align}
\norm*{\mathcal{H}}
&\le
\sum_{\alpha=1}^{3}
\norm*{B_{\mathrm{cell}}^{-\frac12}A^{(\alpha)}B_{\mathrm{cell}}^{-\frac12}}
\norm*{D^{(\alpha)}}\notag\\
&\le
\sum_{\alpha=1}^{3}
\left(\rho^{-\frac12}\norm*{S_{\mathrm{comp}}^{-\frac12}}\right)\frac{1}{h}\notag\\
&=
\frac{3}{h}\rho^{-\frac12}\norm*{S_{\mathrm{comp}}^{-\frac12}}.
\label{eq:H-bound-by-params}
\end{align}
In the second inequality, we used 
\(\norm*{B_{\mathrm{cell}}^{-\frac12}A^{(\alpha)}B_{\mathrm{cell}}^{-\frac12}}
\le \rho^{-\frac12}\norm*{S_{\mathrm{comp}}^{-\frac12}}\)
(from Lemma~\ref{lemma:BAB-bounded}) and
\(\norm*{D^{(\alpha)}}\le 1/h\) (from Ref.~\cite{LectureNotesQuantumAlgorithmsScientificComputation2022}, Proposition~4.12).

Now we apply Proposition~\ref{prop:second-order-trotter-error} with $R=3\cdot 16n=48n$ and with $U_2(\tau)$ defined in Eq.~\eqref{eq:U2}.
Substituting Eq.~\eqref{eq:H-bound-by-params} into Eq.~\eqref{eq:general-norm-scaling-bound}, we obtain
\begin{align}
\norm*{U_2(\tau)-e^{-i\mathcal{H}\tau}}
\le
\frac{2\cdot 1440^{3}}{h^3}\,
\rho^{-\frac32}\,
\norm*{S_{\mathrm{comp}}^{-1}}^{\frac32}\,
n^3\,\tau^3,
\end{align}
where we used
\(\norm*{S_{\mathrm{comp}}^{-\frac12}}^3
=
\norm*{S_{\mathrm{comp}}^{-1}}^{\frac32}\).

The applicability condition of Proposition~\ref{prop:second-order-trotter-error} becomes
\begin{align}
\frac{1440^{3}}{h^3}\,
\rho^{-\frac32}\,
\norm*{S_{\mathrm{comp}}^{-1}}^{\frac32}\,
n^3\,\tau^3
\le 1.
\end{align}
Therefore Eq.~\eqref{eq:2nd-trotter-error} follows.
\end{proof}
As in the first-order case, we move from one step error bound to one step implementation cost.
\begin{lemma}[One step implementation cost for $U_2$]\label{lemma:one-step-U2}
	Let $U_2(\tau)$ be the second-order Trotter operator defined in Eq.~\eqref{eq:U2}.
	Assume $U_2(\tau)$ is implemented by the circuit in Eq.~\eqref{eq:U2-implement}.
	Then one application of $U_2(\tau)$ uses
	\begin{align}
		3n+4
	\end{align}
	qubits and at most
	\begin{align}
		2(432n^2+378)
	\end{align}
	CNOT gates.
\end{lemma}
\begin{proof}
The quantum circuit $U_2(\tau)$ given in Eq.~\eqref{eq:U2} for the second-order Trotter formula consists of two parts. One is the same circuit as the case of the first order Trotter formula for time $\tau/2$, i.e., $U_1(\tau/2)$ given by Eq.~\eqref{eq:U1-circuit}. The other is constructed as the reverse-order product of $U_1(\tau/2)$.
Therefore, the required CNOT gate count is twice the CNOT gate count that is needed for the first-order Trotter formula, so that each Trotter step needs $2(432n^2+378)$ CNOT gates.
\end{proof}
Combining Proposition~\ref{thm:2nd-trotter-error}, Lemma~\ref{lemma:one-step-U2}, and Lemma~\ref{lemma:error-to-gate-conversion}, we obtain the following second-order CNOT gate complexity bound.
\begin{theorem}(Second-order Trotter gate complexity: norm scaling)\label{thm:2nd-trotter-gates}
	Consider the Hamiltonian $\mathcal{H}$ in Eq.~\eqref{eq:H-decomposition-with-48n-components}.
	Using the second-order Trotter formula, the time evolution $e^{-i\mathcal{H}T}$ can be approximated within additive error $\epsilon$ (in operator norm) by a circuit using
	\begin{align}
		3n+4
	\end{align}
	qubits and at most
    \begin{align}\label{eq:2nd-trotter-gates}
	        \frac{2880^{\frac{3}{2}}T^{\frac{3}{2}}}{h^{\frac{3}{2}}\epsilon^{\frac{1}{2}}} \rho^{-\frac{3}{4}} \norm*{S_{\mathrm{comp}}^{-1}}^{\frac{3}{4}} n^{\frac{3}{2}}(432n^2+378).
    \end{align}
		Here, $2^n$ is the number of grid points per coordinate, $h$ is the grid interval, $\rho$ is the mass density, and $S_{\mathrm{comp}}$ is the compliance matrix.
\end{theorem}
\begin{proof}
By Proposition~\ref{thm:2nd-trotter-error},
\begin{align}
	\norm*{e^{-i\mathcal{H}\tau}-U_2(\tau)}\le C\tau^3, \\
	C = \frac{2\cdot 1440^{3}}{h^3} \rho^{-\frac{3}{2}} \norm*{S_{\mathrm{comp}}^{-1}}^{\frac{3}{2}}n^3.
\end{align}
By Lemma~\ref{lemma:one-step-U2}, one step uses $3n+4$ qubits and at most \(2\left(432n^2+378\right)\) CNOT gates.
Applying Lemma~\ref{lemma:error-to-gate-conversion} with \(N_{\mathrm{step}}=2\left(432n^2+378\right)\), we obtain Eq.~\eqref{eq:2nd-trotter-gates}.
\end{proof}
Compared with Eq.~\eqref{eq:1st-trotter-gates} in Theorem~\ref{thm:1st-trotter-error-norm-scaling}, Eq.~\eqref{eq:2nd-trotter-gates} improves the dependence on $T,\epsilon$ from $T^2/\epsilon$ to $T^{3/2}/\epsilon^{1/2}$, and the leading $n$-dependence from $O(n^4)$ to $O(n^{7/2})$.

\subsubsection{Comparison with classical time integration of the same semidiscrete system}
\label{sec:comparison-classical-same-ode}

To make the comparison as direct as possible, we compare the quantum algorithm with a classical explicit integrator applied to the same semidiscrete ODE defined by Eq.~\eqref{eq:discretize-ew}, rather than with a different PDE-level discretization. Since the seven zero-padded components are dynamically decoupled, it suffices to restrict the classical discussion to the physical nine-component sector. In Sec.~\ref{sec:preliminaries}, the corresponding semidiscrete quantum states were denoted by \(\ket*{\bm{w}_h(t)}\) and \(\ket*{\tilde{\bm{u}}_h(t)}\). In the present classical discussion, we use the same symbols without ket notation, namely \(\bm{w}_h(t)\) and \(\tilde{\bm{u}}_h(t)\), for the underlying coefficient vectors. Let \(N:=2^n\) denote the number of grid points per coordinate, so that \(N=\Theta(L/h)\). Removing the decoupled padding sector from Eqs.~\eqref{eq:A-definition} and \eqref{eq:B}, define
\begin{align}
\widetilde{A}
&:=
\sum_{\alpha=1}^{3}\widetilde{A}^{(\alpha)}\otimes D^{(\alpha)},
\qquad
\widetilde{B}
:=
\widetilde{B}_{\mathrm{cell}}\otimes I^{\otimes 3n},
\\
\widetilde{A}^{(\alpha)}
&:=
\begin{bmatrix}
O_{3\times 3} & C^{(\alpha)}\\
(C^{(\alpha)})^{\mathsf T} & O_{6\times 6}
\end{bmatrix},
\qquad
\widetilde{B}_{\mathrm{cell}}
:=
\begin{bmatrix}
\rho I_{3\times 3} & O_{3\times 6}\\
O_{6\times 3} & S_{\mathrm{comp}}
\end{bmatrix}.
\end{align}
The corresponding classical semidiscrete evolution is
\begin{align}
\dv{t}\bm{w}_h(t)=\widetilde{B}^{-1}\widetilde{A}\bm{w}_h(t),
\quad
\bm{w}_h(0)=\bm{w}_h^0,
\label{eq:classical-same-ode-main}
\end{align}
with \(\bm{w}_h(t)\in\mathbb{C}^{9N^3}\). Introducing the transformed variables
\begin{align}
\tilde{\bm{u}}_h(t):=\widetilde{B}^{1/2}\bm{w}_h(t),
\quad
\tilde{\bm{u}}_h(0)=:\tilde{\bm{u}}_h^0=\widetilde{B}^{1/2}\bm{w}_h^0,
\end{align}
we obtain
\begin{align}
\dv{t}\tilde{\bm{u}}_h(t)=K_h\tilde{\bm{u}}_h(t),
\label{eq:classical-Kh-main}
\end{align}
where
\begin{align}
K_h
&:=\widetilde{B}^{-1/2}\widetilde{A}\widetilde{B}^{-1/2}
=
\begin{bmatrix}
O & L_h\\
-L_h^* & O
\end{bmatrix},\\
L_h
&:=
\sum_{\alpha=1}^{3}
\rho^{-1/2}C^{(\alpha)}S_{\mathrm{comp}}^{-1/2}\otimes D^{(\alpha)}.
\label{eq:classical-Lh-main}
\end{align}

As a classical integrator for Eq.~\eqref{eq:classical-Kh-main}, we consider the partitioned leapfrog method, equivalently the St\"ormer--Verlet method, or Strang splitting,
\begin{align}
\Psi_\tau
:=
e^{\frac{\tau}{2}K_{1,h}}
e^{\tau K_{2,h}}
e^{\frac{\tau}{2}K_{1,h}},
\label{eq:classical-Psi-main}
\end{align}
with
\begin{align}
K_{1,h}
:=
\begin{bmatrix}
O & L_h\\
O & O
\end{bmatrix},
\qquad
K_{2,h}
:=
\begin{bmatrix}
O & O\\
-L_h^* & O
\end{bmatrix}.
\end{align}
Since \(K_{1,h}^2=K_{2,h}^2=O\), each split step is evaluated exactly. Writing
\begin{align}
\tilde{\bm{u}}_h(t)=
\begin{bmatrix}
\bm{q}(t)\\
\bm{r}(t)
\end{bmatrix},
\qquad
\bm{q}(t)\in\mathbb{C}^{3N^3},
\quad
\bm{r}(t)\in\mathbb{C}^{6N^3},
\end{align}
and, for discrete times \(t_m:=m\tau\), write
\begin{align}
\tilde{\bm{u}}_h^m:=\tilde{\bm{u}}_h(t_m)
=
\begin{bmatrix}
\bm{q}^m\\
\bm{r}^m
\end{bmatrix}.
\end{align}
Then one time step is
\begin{align}
\bm{q}^{m+\frac12}
&=
\bm{q}^m+\frac{\tau}{2}L_h\bm{r}^m,
\\
\bm{r}^{m+1}
&=
\bm{r}^m-\tau L_h^*\bm{q}^{m+\frac12},
\\
\bm{q}^{m+1}
&=
\bm{q}^{m+\frac12}+\frac{\tau}{2}L_h\bm{r}^{m+1}.
\label{eq:classical-update-main}
\end{align}

The detailed analysis is given in Appendix~\ref{app:classical-same-ode}. There we show that if
\begin{align}
\tau\|L_h\|\le \eta<2,
\end{align}
then \(\Psi_\tau\) is power-bounded:
\begin{align}
\|\Psi_\tau^m\|\le C_\eta,
\qquad
C_\eta:=\left(1-\frac{\eta^2}{4}\right)^{-1/2}.
\label{eq:classical-stability-main}
\end{align}
Moreover, if \(T=M\tau\) and \(\tau\|L_h\|\le \min\{1,\eta\}\), then the final-time operator-norm error satisfies
\begin{align}
\|e^{T K_h}-\Psi_\tau^M\|
\le
\frac{C_\eta}{2}T\tau^2\|L_h\|^3.
\label{eq:classical-global-error-main}
\end{align}
Therefore, \(\|e^{T K_h}-\Psi_\tau^M\|\le \epsilon\) is guaranteed by
\begin{align}
\tau
\le
\min\left\{
\frac{\eta}{\|L_h\|},
\sqrt{\frac{2\epsilon}{C_\eta T\|L_h\|^3}}
\right\},
\label{eq:classical-step-condition-main}
\end{align}
which yields
\begin{align}
m_{\mathrm{cl}}
=
\mathcal{O}\!\left(
\max\left\{
T\|L_h\|,
\sqrt{\frac{T^3\|L_h\|^3}{\epsilon}}
\right\}
\right)
\label{eq:classical-step-count-main}
\end{align}
time steps.

To convert this into a parameter-dependent bound, note that
\begin{align}
\|L_h\|
&\le
\sum_{\alpha=1}^{3}
\rho^{-1/2}\|C^{(\alpha)}\|\,
\|S_{\mathrm{comp}}^{-1/2}\|\,
\|D^{(\alpha)}\|\\
&\le
\frac{3}{h}\rho^{-1/2}\|S_{\mathrm{comp}}^{-1/2}\|\\
&=
\frac{3v}{h},
\label{eq:classical-Lh-bound-main}
\end{align}
where
\begin{align}
v:=\sqrt{\rho^{-1}\norm*{S_{\mathrm{comp}}^{-1}}}.
\end{align}
Furthermore, each time step applies \(L_h\) twice and \(L_h^*\) once, together with vector additions and scalar multiplications, so the arithmetic cost per step is \(\mathcal{O}(N^3)\). Consequently,
\begin{align}
W_{\mathrm{cl}}
&=
\mathcal{O}\!\left(
N^3
\max\left\{
\frac{Tv}{h},
\frac{T^{3/2}v^{3/2}}{h^{3/2}\epsilon^{1/2}}
\right\}
\right),
\label{eq:classical-total-cost-main}
\\
M_{\mathrm{cl}}
&=
\mathcal{O}(N^3).
\end{align}

Using Theorems~\ref{thm:1st-trotter-gates} and \ref{thm:2nd-trotter-gates} for the quantum algorithm, we obtain the asymptotic comparison in Table~\ref{tab:complexity}. The main structural gap remains exponential in spatial resolution: the classical integrator stores \(\mathcal{O}(N^3)\) degrees of freedom, whereas the quantum algorithm uses \(\mathcal{O}(\log N)\) qubits. 
Also, given that $N = \Theta (L / h)$, our quantum algorithm achieves a nearly cubic improvement in the bound on time complexity with respect to $N$, which suggests the potential for polynomial quantum speedup.
Note that this is a comparison based on the same semidiscretization rather than a full PDE-level comparison. Spatial discretization error, source terms, state preparation, readout, and fault-tolerance overhead are excluded.
\begin{table}[t]
\centering
\caption{Asymptotic complexity comparison for the same semidiscrete system. Here \(N=2^n=\Theta(L/h)\) is the number of grid points per coordinate, \(T\) is the simulation time, \(h\) is the grid spacing, \(v=\sqrt{\rho^{-1}\norm*{S_{\mathrm{comp}}^{-1}}}\), and \(\epsilon\) is the target operator-norm error bound for both the classical and quantum simulations.}
\label{tab:complexity}
\scriptsize
\setlength{\tabcolsep}{3pt}
\renewcommand{\arraystretch}{1.1}
\resizebox{\columnwidth}{!}{%
\begin{tabular}{@{}llcc@{}}
\toprule
Category & Method & Memory & Time \\
\midrule
Classical
& Partitioned leapfrog / Strang
& $\mathcal{O}(N^3)$
& $\mathcal{O}\!\left(
N^3
\max\left\{
\frac{Tv}{h},
\frac{T^{3/2}v^{3/2}}{h^{3/2}\epsilon^{1/2}}
\right\}
\right)$ \\
\midrule
\multirow{2}{*}{Quantum}
& 1st-order Trotter
& $\mathcal{O}(\log N)$
& $\mathcal{O}\!\left(
(\log N)^3\frac{T^2v^2}{h^2\epsilon}
\right)$ \\
& 2nd-order Trotter
& $\mathcal{O}(\log N)$
& $\mathcal{O}\!\left(
(\log N)^{7/2}\frac{T^{3/2}v^{3/2}}{h^{3/2}\epsilon^{1/2}}
\right)$ \\
\bottomrule
\end{tabular}%
}
\end{table}

\section{Numerical experiments}\label{sec:numerical-experiments}
In this section, we numerically validate that the proposed quantum circuit based on the first-order Trotter formula approximates the target real-time evolution generated by the Hamiltonian $\mathcal{H}$.
As a reference, we compute the exact evolution by direct matrix exponentiation and compare it with the quantum state produced by the Trotterized circuit.
We examine this agreement at two levels: the quantum state level through fidelity, and the physical field level through reconstructed variables $v_z$ and $\sigma_{zz}$.

The circuits were implemented in Qiskit (1.3.0)~\cite{qiskit2024}.
We consider the semidiscrete first-order formulation of the three-dimensional elastic wave equation introduced in Sec.~\ref{sec:preliminaries}.
Through this section, we assume a linear isotropic medium with mass density $\rho=1$, Young's modulus $E=0.646$, and Poisson ratio $\nu=0.255$.
We set total simulation time to $T=30$, the grid spacing to $h=1$, and the number of qubits used for each spatial coordinate to $n=5$.
Hence, the number of grid points per coordinate is $N=2^n$.

Define
\begin{align}
I_{\mathrm{c}} &:= \left\{\frac{N}{2}-1,\frac{N}{2}\right\}, &
I_{\mathrm{b}} &:= \left\{2,3,\dots,N-3\right\}.
\end{align}
We use the following initial states. The subscript ``pulse'' denotes a localized excitation supported on the central \(2\times 2\times 2\) grid block, while ``p'' and ``s'' denote initial states supported on the mid-plane \(j_x\in I_{\mathrm{c}}\), representing longitudinal (P-wave) and transverse (S-wave) excitations, respectively.
\begin{align}
\ket*{\bm{w}_{\mathrm{pulse}}}
&=
\frac{1}{C_{\mathrm{pulse}}}
\sum_{j_x\in I_{\mathrm{c}}}
\sum_{j_y\in I_{\mathrm{c}}}
\sum_{j_z\in I_{\mathrm{c}}}
\ket*{2}\ket*{j_x}\ket*{j_y}\ket*{j_z},
\label{eq:pulse-wave}\\
\ket*{\bm{w}_{\mathrm{p}}}
&=
\frac{1}{C_{\mathrm{p}}}
\sum_{j_x\in I_{\mathrm{c}}}
\sum_{j_y\in I_{\mathrm{b}}}
\sum_{j_z\in I_{\mathrm{b}}}
\ket*{2}\ket*{j_x}\ket*{j_y}\ket*{j_z},
\label{eq:p-wave}\\
\ket*{\bm{w}_{\mathrm{s}}}
&=
\frac{1}{C_{\mathrm{s}}}
\sum_{j_x\in I_{\mathrm{c}}}
\sum_{j_y\in I_{\mathrm{b}}}
\sum_{j_z\in I_{\mathrm{b}}}
\ket*{0}\ket*{j_x}\ket*{j_y}\ket*{j_z},
\label{eq:s-wave}
\end{align}
where \(C_{\mathrm{pulse}}, C_{\mathrm{p}}, C_{\mathrm{s}} \in \mathbb{C}\) are normalization constants.

For a fixed Trotter step size $\tau$, we compare the exact state
\begin{align}
\ket{\psi_{\mathrm{ref}}(t_m)}
=
e^{-i\mathcal{H}t_m}\ket{\psi_0},
\qquad
t_m:=m\tau,
\end{align}
with the Trotterized state
\begin{align}
\ket{\psi_{\mathrm{T}}(t_m)}
=
\left(U_1(\tau)\right)^m\ket{\psi_0},
\end{align}
where $m=0,1,\ldots,T/\tau$.
In fidelity comparison below, we use the step sizes
\begin{align}
\tau \in \{0.1,\,0.2,\,0.5,\,1.0\},
\end{align}
so that $T/\tau$ is an integer in every case.

As a state-level metric, we use the fidelity
\begin{align}
F(t_m)
=
\left|\braket{\psi_{\mathrm{ref}}(t_m)}{\psi_{\mathrm{T}}(t_m)}\right|^2.
\end{align}

\begin{figure*}[t]
\centering
\subfloat[Different Trotter step sizes\label{fig:fidelity_time_intervals}]{{\includegraphics[width=0.47\textwidth]{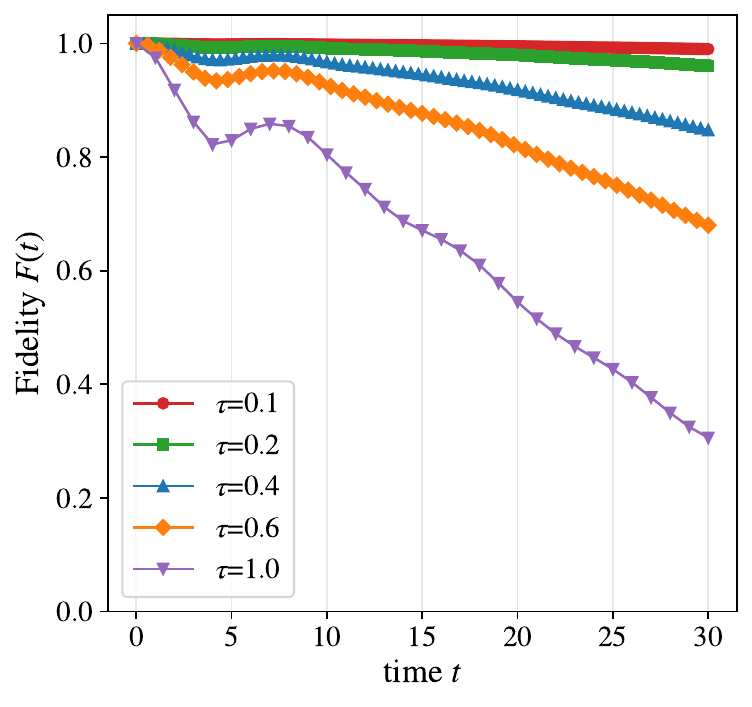}}}
\hfill
\subfloat[Different initial states ($\tau=0.1$)\label{fig:fidelity_3initstate}]{{\includegraphics[width=0.47\textwidth]{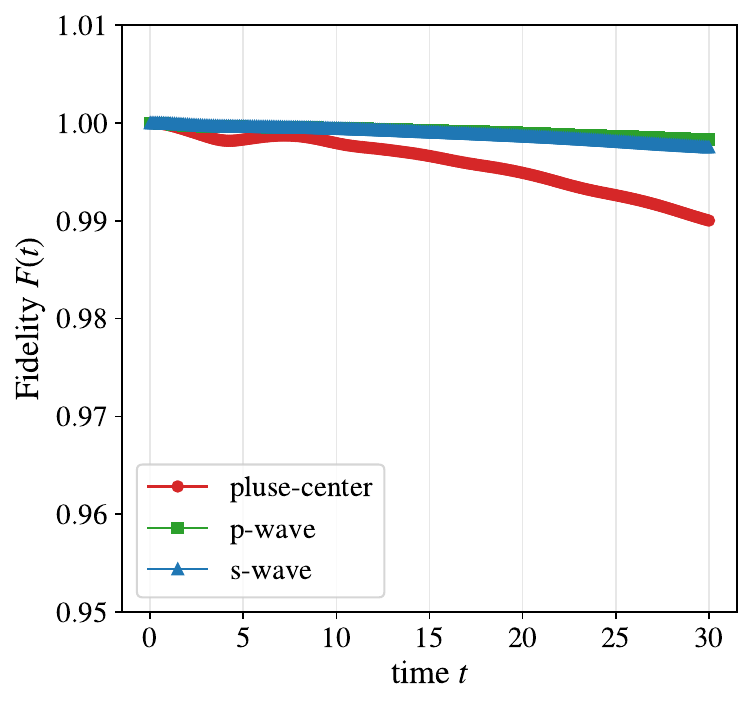}}}
\caption{Fidelity \(F(t)\) between the exact state and the Trotterized state.}
\label{fig:fidelity_time_intervals_3initstate}
\end{figure*}

Figure~\ref{fig:fidelity_time_intervals_3initstate} presents the state-level comparison through fidelity.
In Fig.~\ref{fig:fidelity_time_intervals}, we fix the initial state to the center pulse state in Eq.~\eqref{eq:pulse-wave} and vary the Trotter step size.
Since the only approximation here is the Trotter formula, the fidelity gets closer to $1$ as $\tau$ becomes smaller.

In Fig.~\ref{fig:fidelity_3initstate}, we fix the Trotter step size to $\tau=0.1$ and compare three initial states in Eqs.~\eqref{eq:pulse-wave}--\eqref{eq:s-wave}.
The fidelity remains close to $1$ in all three cases, while the center pulse state yields a slightly smaller fidelity than the primary wave and secondary wave initial states.
Since our analysis is based on the operator norm bounds, this difference is not reflected in the theoretical worst case estimate; rather it illustrates state dependent behavior observed in these numerical examples.

As a field level comparison, we map both states back to the physical variables through the inverse transformation associated with the Schr\"{o}dinger-form formulation, and compare the reconstructed fields $v_z$ and $\sigma_{zz}$.

\begin{figure*}[t]
	\centering
	\subfloat[$v_z$ (2D)\label{fig:vz_2d}]{{\includegraphics[width=0.84\textwidth]{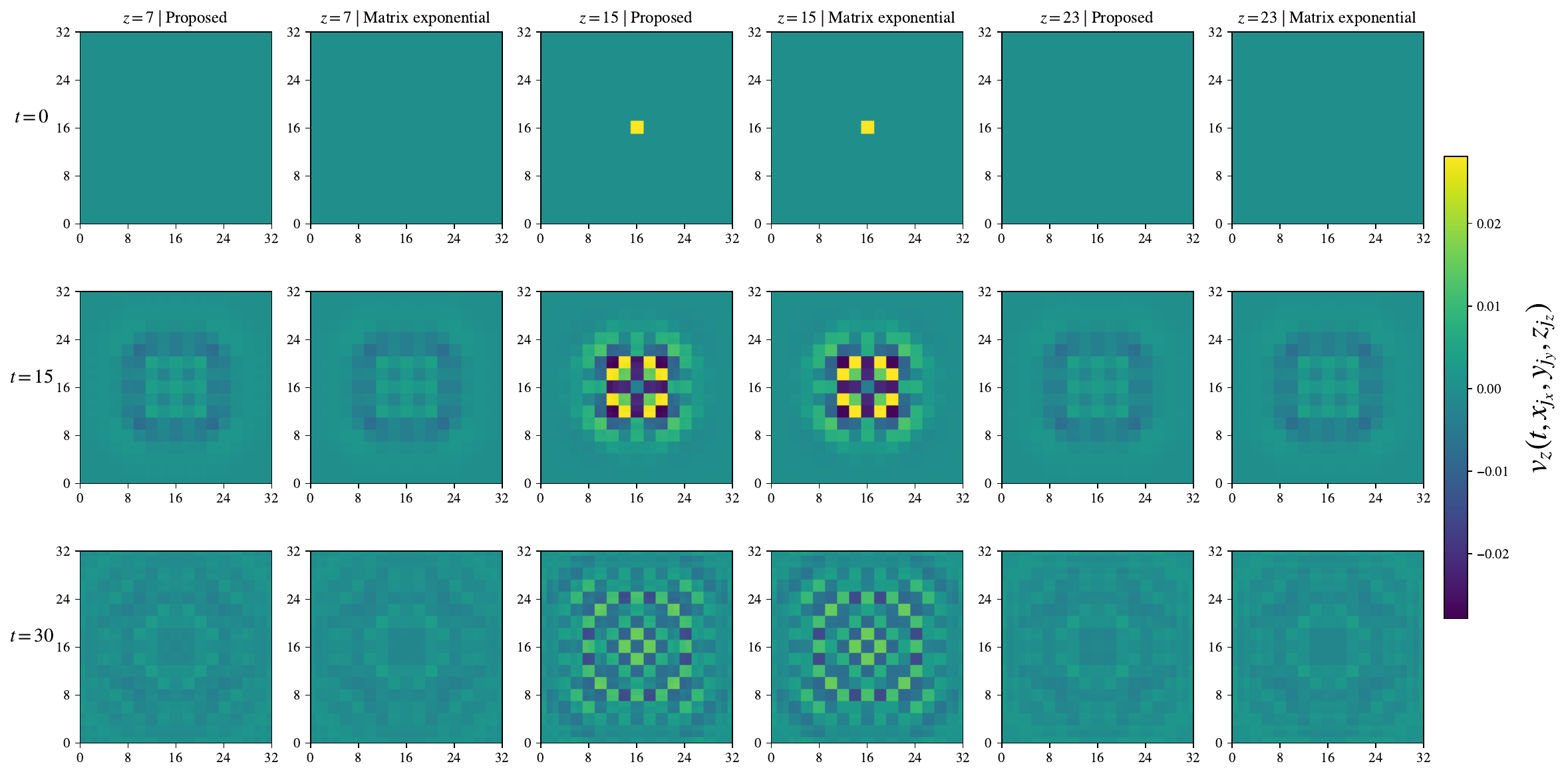}}}\\[0.4em]
	\subfloat[$\sigma_{zz}$ (2D)\label{fig:sigma_zz_2d}]{{\includegraphics[width=0.84\textwidth]{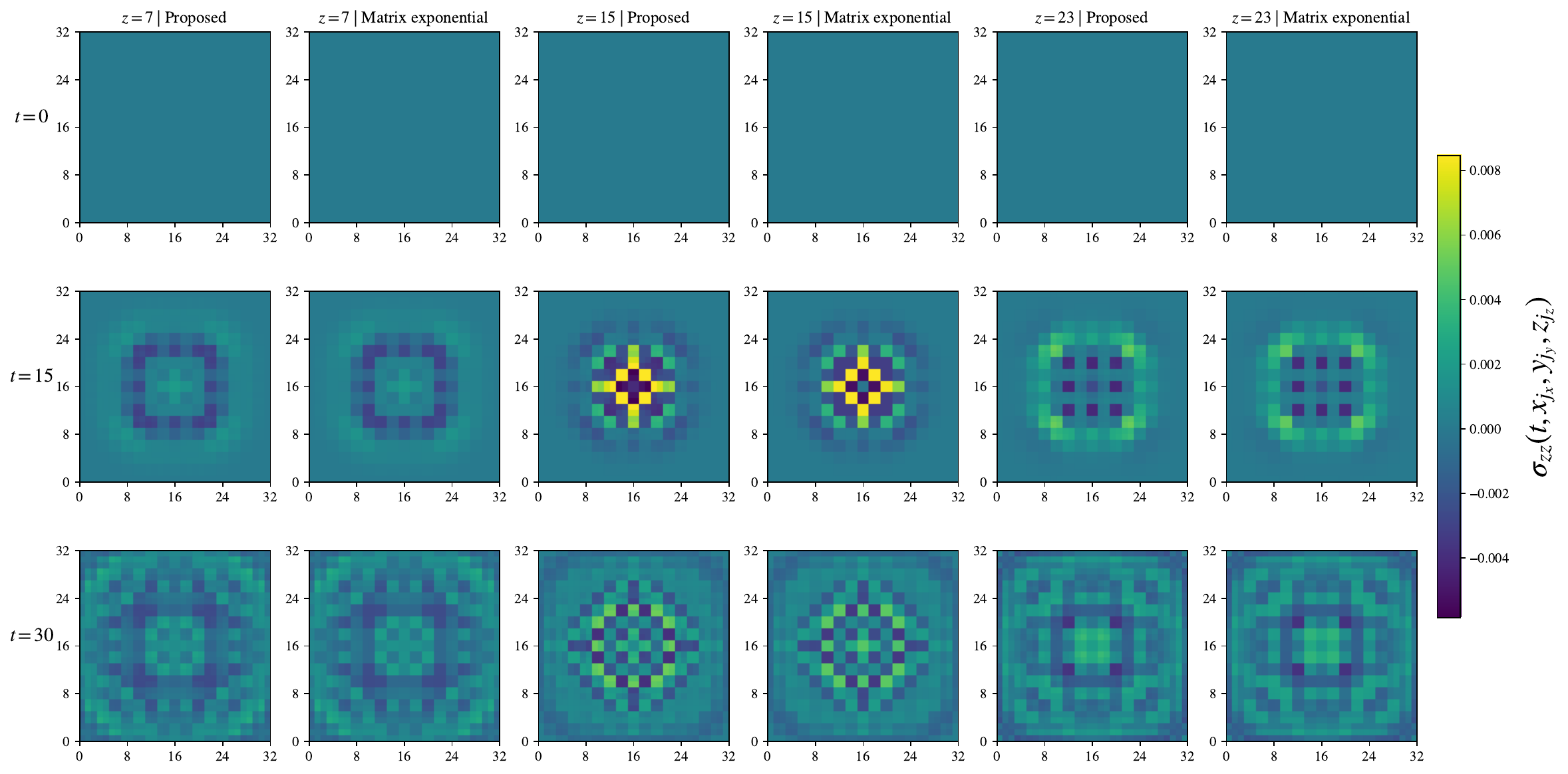}}}
		\caption{Comparison of the reconstructed fields \(v_z\) and \(\sigma_{zz}\) between the exact evolution and the first-order Trotter circuit for the center pulse initial state with \(\tau=0.1\). For visualization, the top 2\% of values are clipped so that lower amplitude structures are more clearly visible.}
	\label{fig:vz_sigma_zz_2d}
\end{figure*}

Figure~\ref{fig:vz_sigma_zz_2d} compares the reconstructed fields $v_z$ and $\sigma_{zz}$ obtained from the exact evolution and from the proposed Trotter circuit.
The snapshots are shown on a two-dimensional slice for several representative times under the center pulse initial state in Eq.~\eqref{eq:pulse-wave}.
The two sets of fields are in close agreement over the entire time interval, indicating that the proposed circuit reproduces not only the quantum state itself but also the physically relevant observables derived from it.

Taken together, the fidelity results in Fig.~\ref{fig:fidelity_time_intervals_3initstate} and the field-level agreement in Fig.~\ref{fig:vz_sigma_zz_2d} support the conclusion that the proposed quantum circuit provides an accurate approximation of the target real time evolution for the parameter range considered here.

\section{Conclusion}\label{sec:conclusion}
In this paper, we proposed an explicit quantum circuit implementation for simulating a first-order formulation of the three-dimensional elastic wave equation.
By rewriting the discretized elastic wave dynamics in Schr\"{o}dinger form, and applying an eigendecomposition to the finite dimensional coefficient matrix acting on the velocity and stress component space, we obtained a Hamiltonian representation that is compatible with tensor product based circuit construction.
Based on this representation, we construct first-order and second-order Trotter circuits for the time evolution operator.

We then analyzed the approximation error and gate complexity of the proposed circuits.
For the first-order Trotter formula, we derived both a norm scaling bound and a sharper commutator scaling bound, and showed that the latter improves the leading $n$-dependence of the CNOT complexity.
For the second-order Trotter formula, we obtained the corresponding one step error bound and global gate complexity estimate.
These results give explicit gate complexity bounds in terms of the discretization parameter, the simulation time, the target accuracy, and the material parameters.

We also compared the asymptotic resource scaling of the proposed quantum algorithm with that of a classical finite difference method.
While the classical method requires $\mathcal{O}(N^3)$ memory for a grid with $N$ points per coordinate, the quantum representation uses only $\mathcal{O}(\log{N})$ qubits.
The quantum time complexity is expressed in terms of circuit-level CNOT counts, with explicit dependence on $T,\epsilon$, and the material parameters through the Trotter error analysis.

Finally, numerical experiments confirmed that the proposed first-order Trotter circuit accurately approximates the exact time evolution for the parameter range considered here.
At the state-level, the fidelity with the exact solution remained close to 1, and at the field-level, the reconstructed quantities $v_z$ and $\sigma_{zz}$ showed close agreement with the reference solution.
These results support the validity of the proposed approach both theoretically and numerically.

Several extensions remain for future work.
One direction is to investigate whether the present structured decomposition, which separates the low dimensional component space from the discretized spatial degrees of freedom, can be adapted to an LCHS formulation.
Another independent direction for future work is to extend the present framework to spatially varying material parameters and more general geometries.
For sufficiently simple spatial dependence, such effects may be incorporated by slightly enlarging the finite dimensional register for the local component degrees of freedom, while preserving the separation from the spatial register induced by discretization.

\section*{Acknowledgements}
We thank Katsuhiro Endo, Tadashi Kadowaki, Shunya Minami, and Yohichi Suzuki for helpful discussions. 
A part of this work was performed for Council for Science, Technology and Innovation (CSTI),
Cross-ministerial Strategic Innovation Promotion Program (SIP), ``Promoting the application of advanced
quantum technology platforms to social issues''(Funding agency : QST).
\bibliography{ref}
\clearpage
\onecolumngrid
\appendix
\section{Auxiliary Lemmas and Deferred Proof for the First-Order Commutator Bound}\label{app:first-order-commutator}
In the following, $\trnorm{\cdot}$ denotes trace norm as $\trnorm{A}:= \tr \left[\sqrt{A^{\dagger}A}\right]$.
\begin{lemma}\label{thm:property-of-block-matrix}
    Let $G$ be an operator and define
    \begin{align}
	F =
	\begin{bmatrix}
		O & G\\
		G^{\dagger} & O
	\end{bmatrix},
    \end{align}
    where $G^{\dagger}$ denotes the adjoint of $G$.
    Then, the following relations hold:
    \begin{align}
	&\norm*{F}
	\leq
	\max\left\{\norm*{G},\norm*{G^{\dagger}}\right\},\label{eq:norm1}\\
	&\trnorm{F}
	= \trnorm{G}+\trnorm{G^{\dagger}},\label{eq:norm2}\\
	&\trnorm{F^2}
	=
	2\trnorm{G^{\dagger}G}.\label{eq:norm3}
    \end{align}
\end{lemma}
\begin{proof}
Eq. (\ref{eq:norm1}) follows straightforwardly.
For Eq. (\ref{eq:norm2}), the block diagonal structure of $F^{\dagger}F$ ensures that
\begin{align}
		F^{\dagger}F
		=
		\begin{bmatrix}
			GG^{\dagger} & O\\
			O & G^{\dagger}G
		\end{bmatrix}.
\end{align}
Taking the trace yields the desired result.
Similarly, for Eq. (\ref{eq:norm3}), $\sqrt{\left(F^2\right)^{\dagger}\left(F^2\right)}$ have a block diagonal form:
\begin{align}
    \sqrt{\left(F^2\right)^{\dagger}\left(F^2\right)}
    =
    \begin{bmatrix}
        \sqrt{GG^{\dagger}GG^{\dagger}} & O\\
        O & \sqrt{G^{\dagger}GG^{\dagger}G}
    \end{bmatrix}.
\end{align}
Taking the trace norm yields
\begin{align}
    \trnorm{F^2}
    &=
    \tr \left[\sqrt{\left(F^2\right)^{\dagger}\left(F^2\right)}\right]\notag\\
    &=
    \tr\left[\sqrt{GG^{\dagger}GG^{\dagger}}\right] + \tr\left[\sqrt{G^{\dagger}GG^{\dagger}G}\right]\notag\\
    &=
    \tr\left[GG^{\dagger}\right] + \tr\left[G^{\dagger}G\right]\notag\\
    &=
    2\tr\left[G^{\dagger}G\right]\notag\\
    &=
    2\tr\left[\sqrt{G^{\dagger}GG^{\dagger}G}\right]\notag\\
    &=
    2\trnorm{G^{\dagger}G}.
\end{align}
\end{proof}
\begin{lemma}\label{lemma:BAB-bounded}
Let $A^{(\alpha)}$ and $B$ be defined by Eq.~\eqref{eq:A-definition} and Eq.~\eqref{eq:B}, respectively.
For each $\alpha$, the following bounds hold:
\begin{align}
	\norm\Big{B_{\mathrm{cell}}^{-\frac{1}{2}}A^{(\alpha)}B_{\mathrm{cell}}^{-\frac{1}{2}}}
	&\leq
	\rho^{-\frac{1}{2}}\norm\Big{\mathcal{S}^{-\frac{1}{2}}_{\mathrm{comp}}},\notag\\
		\trnormBig{B_{\mathrm{cell}}^{-\frac{1}{2}}A^{(\alpha)}B_{\mathrm{cell}}^{-\frac{1}{2}}}
	&=
	6\rho^{-\frac{1}{2}}\norm\Big{\mathcal{S}^{-\frac{1}{2}}_{\mathrm{comp}}},\notag\\
		\trnormBig{\left(B_{\mathrm{cell}}^{-\frac{1}{2}}A^{(\alpha)}B_{\mathrm{cell}}^{-\frac{1}{2}}\right)^2}
	&\leq
	6\rho^{-1}\norm\Big{\mathcal{S}^{-1}_{\mathrm{comp}}}\label{eq:normtri3}.
\end{align}
where $\rho$ is the mass density and $S_{\mathrm{comp}}$ is the compliance matrix.
\end{lemma}
\begin{proof}
We first rewrite
\begin{align}
	B_{\mathrm{cell}}^{-\frac{1}{2}}A^{(\alpha)}B_{\mathrm{cell}}^{-\frac{1}{2}}
	=
	\begin{bmatrix}
		\widetilde{B}_{\mathrm{cell}}^{-\frac{1}{2}}\widetilde{A}^{(\alpha)}\widetilde{B}_{\mathrm{cell}}^{-\frac{1}{2}} & O\\
		O & O
	\end{bmatrix},
\end{align}
where
\begin{align}
	\widetilde{B}_{\mathrm{cell}}=
	\begin{bmatrix}
		\rho I_{3\times 3} & O_{3\times 6}\\
		O_{6\times 3} & \mathcal{S}_{\mathrm{comp}}
	\end{bmatrix},
	\qquad
	\widetilde{A}^{(\alpha)}=
	\begin{bmatrix}
		O_{3\times 3} & C^{(\alpha)}\\
		(C^{(\alpha)})^{T} & O_{6 \times 6}
	\end{bmatrix}.
\end{align}
Hence
\begin{align}
	\norm\Big{B_{\mathrm{cell}}^{-\frac{1}{2}}A^{(\alpha)}B_{\mathrm{cell}}^{-\frac{1}{2}}}
	&=
	\norm\Big{\widetilde{B}_{\mathrm{cell}}^{-\frac{1}{2}}\widetilde{A}^{(\alpha)}\widetilde{B}_{\mathrm{cell}}^{-\frac{1}{2}}},\notag\\
	\trnormBig{B_{\mathrm{cell}}^{-\frac{1}{2}}A^{(\alpha)}B_{\mathrm{cell}}^{-\frac{1}{2}}}
	&=
	\trnormBig{\widetilde{B}_{\mathrm{cell}}^{-\frac{1}{2}}\widetilde{A}^{(\alpha)}\widetilde{B}_{\mathrm{cell}}^{-\frac{1}{2}}},\notag\\
	\trnormBig{\left(B_{\mathrm{cell}}^{-\frac{1}{2}}A^{(\alpha)}B_{\mathrm{cell}}^{-\frac{1}{2}}\right)^2}
	&=
	\trnormBig{\left(\widetilde{B}_{\mathrm{cell}}^{-\frac{1}{2}}\widetilde{A}^{(\alpha)}\widetilde{B}_{\mathrm{cell}}^{-\frac{1}{2}}\right)^2}.
\label{eq:tilde-onoff-trnorm2}
\end{align}
By direct multiplication,
\begin{align}
	\widetilde{B}_{\mathrm{cell}}^{-\frac{1}{2}}\widetilde{A}^{(\alpha)}\widetilde{B}_{\mathrm{cell}}^{-\frac{1}{2}}
	=
	\begin{bmatrix}
		O & G\\
		G^{\dagger} & O
	\end{bmatrix},
	\qquad
	G=\rho^{-\frac{1}{2}}C^{(\alpha)}\mathcal{S}_{\mathrm{comp}}^{-\frac{1}{2}}.
\end{align}
Applying Lemma~\ref{thm:property-of-block-matrix}, together with
\(\trnorm{AB}\le \norm*{A}\trnorm{B}\) and \(\trnorm{AB}\le \trnorm{A}\norm*{B}\), we obtain
\begin{align}
	\norm\Big{\widetilde{B}_{\mathrm{cell}}^{-\frac{1}{2}}\widetilde{A}^{(\alpha)}\widetilde{B}_{\mathrm{cell}}^{-\frac{1}{2}}}
	&\le
	\max\left\{\norm\Big{\rho^{-\frac{1}{2}}C^{(\alpha)}\mathcal{S}_{\mathrm{comp}}^{-\frac{1}{2}}},\norm\Big{\mathcal{S}_{\mathrm{comp}}^{-\frac{1}{2}}(C^{(\alpha)})^T\rho^{-\frac{1}{2}}}\right\}\notag\\
	&\le
	\norm\Big{\rho^{-\frac{1}{2}}}\norm\Big{\mathcal{S}_{\mathrm{comp}}^{-\frac{1}{2}}}
	\max\left\{\norm\Big{C^{(\alpha)}},\norm\Big{(C^{(\alpha)})^T}\right\}\notag\\
	&\le
	\rho^{-\frac{1}{2}}\norm\Big{\mathcal{S}_{\mathrm{comp}}^{-\frac{1}{2}}},\notag\\
	\trnormBig{\widetilde{B}_{\mathrm{cell}}^{-\frac{1}{2}}\widetilde{A}^{(\alpha)}\widetilde{B}_{\mathrm{cell}}^{-\frac{1}{2}}}
	&=
	\trnormBig{\rho^{-\frac{1}{2}}C^{(\alpha)}\mathcal{S}_{\mathrm{comp}}^{-\frac{1}{2}}}
	+\trnormBig{\mathcal{S}_{\mathrm{comp}}^{-\frac{1}{2}}(C^{(\alpha)})^{T}\rho^{-\frac{1}{2}}}\notag\\
	&\le
	\norm\Big{\rho^{-\frac{1}{2}}}\norm\Big{\mathcal{S}_{\mathrm{comp}}^{-\frac{1}{2}}}
	\left(\trnormBig{C^{(\alpha)}}+\trnormBig{(C^{(\alpha)})^T}\right)\notag\\
	&\le
	6\rho^{-\frac{1}{2}}\norm\Big{\mathcal{S}_{\mathrm{comp}}^{-\frac{1}{2}}},\notag\\
	\trnormBig{\left(\widetilde{B}_{\mathrm{cell}}^{-\frac{1}{2}}\widetilde{A}^{(\alpha)}\widetilde{B}_{\mathrm{cell}}^{-\frac{1}{2}}\right)^2}
	&=
	2\trnormBig{\mathcal{S}_{\mathrm{comp}}^{-\frac{1}{2}}(C^{(\alpha)})^{T}\rho^{-1}C^{(\alpha)}\mathcal{S}_{\mathrm{comp}}^{-\frac{1}{2}}}\notag\\
	&\le
	2\norm\Big{\rho^{-\frac{1}{2}}}^2
	\norm\Big{\mathcal{S}_{\mathrm{comp}}^{-\frac{1}{2}}}^2
	\norm\Big{C^{(\alpha)}}\trnormBig{(C^{(\alpha)})^{T}}\notag\\
	&\le
	6\rho^{-1}\norm\Big{\mathcal{S}_{\mathrm{comp}}^{-1}}.
\end{align}
Combining this with Eq.~\eqref{eq:tilde-onoff-trnorm2} proves the lemma.
\end{proof}
\begin{proof}[Proof of Proposition~\ref{thm:1st-trotter-error}]	
From Proposition~\ref{prop:frist-order-trotter-error-original}, applying the triangle inequality gives
\begin{align}\label{eq:t-error}
	\norm*{U_1(\tau)- e^{-i\mathcal{H}\tau}}
	\leq
	\frac{\tau^2}{2}
	\left[
	\sum_{\alpha=1}^{3}\Xi_{\alpha}^{\mathrm{same}}
	+
	\sum_{\alpha<\beta}\Xi_{\alpha\beta}^{\mathrm{cross}}
	\right],
\end{align}
where
\begin{align}
\Xi_{\alpha}^{\mathrm{same}}
&:=
\sum_{j=0}^{15}\sum_{k=1}^{n}
\norm\Bigg{\left[H^{(\alpha)}_{jk},
\sum_{l=1}^{15}\sum_{m=1}^{n}H^{(\alpha)}_{lm}
+\sum_{m=k+1}^{n}H^{(\alpha)}_{jm}\right]},\\
\Xi_{\alpha\beta}^{\mathrm{cross}}
&:=
\sum_{j=0}^{15}\sum_{k=1}^{n}
\norm\Bigg{\left[H^{(\alpha)}_{jk},
\sum_{l=0}^{15}\sum_{m=1}^{n}H^{(\beta)}_{lm}\right]}.
\end{align}
We first evaluate $\Xi_{\alpha}^{\mathrm{same}}$.
Using $H_{jk}^{(\alpha)}=P_j^{(\alpha)}\otimes \frac{i\lambda_j^{(\alpha)}}{2h}S_k^{(\alpha)}$
and
\begin{align}
[A_1\otimes B_1, A_2\otimes B_2]
=
[A_1,A_2]\otimes B_1B_2 + A_2A_1\otimes [B_1,B_2],
\end{align}
we obtain
\begin{align}
\left[H^{(\alpha)}_{jk},H^{(\alpha)}_{lm}\right]
=
-\frac{\lambda_j^{(\alpha)}\lambda_l^{(\alpha)}}{4h^2}
\left(
\left[P_j^{(\alpha)},P_l^{(\alpha)}\right]\otimes S_k^{(\alpha)}S_m^{(\alpha)}
+
P_l^{(\alpha)}P_j^{(\alpha)}\otimes\left[S_k^{(\alpha)},S_m^{(\alpha)}\right]
\right).
\end{align}
Since $P_j^{(\alpha)} = \ketbra*{\phi_j^{(\alpha)}}$ are orthogonal rank one projectors, $[P_j^{(\alpha)},P_l^{(\alpha)}]=0$ and $P_j^{(\alpha)}P_l^{(\alpha)} = \delta_{jl}P_j^{(\alpha)}$.
Hence 
\begin{align}
\left[H^{(\alpha)}_{jk},H^{(\alpha)}_{lm}\right]
=
-\delta_{jl}\frac{(\lambda_j^{(\alpha)})^2}{4h^2}
P_j^{(\alpha)}\otimes\left[S_k^{(\alpha)},S_m^{(\alpha)}\right].
\label{eq:h-exchange-same-alpha}
\end{align}
From the structure of $S_k^{(\alpha)}$
(Ref.~\cite{Hamiltoniansimulationhyperbolicpartialdifferentialequationsscalablequantumcircuits2024}, Eq.~(42)), we have
\begin{align}
	[S_k^{(\alpha)}, S_m^{(\alpha)}] = 0,
\end{align}
for $n>m>k>1$, and
\begin{align}
	\norm*{[S_k^{(\alpha)}, S_m^{(\alpha)}]}=1
\end{align}
for $n>m>k=1$, so that the only nonzero contributions are the pair with exactly one of $(k,m)$ equal to $1$, and
\begin{align}
\norm*{[H^{(\alpha)}_{j1},H^{(\alpha)}_{jm}]}
=
\frac{1}{4h^2}(\lambda_j^{(\alpha)})^2
\qquad (m\ge 2).
\end{align}
Therefore
\begin{align}
\Xi_{\alpha}^{\mathrm{same}}
&=
\sum_{j=0}^{15}\sum_{k=1}^{n}
\norm\Bigg{\left[H^{(\alpha)}_{jk},
\sum_{l=1}^{15}\sum_{m=1}^{n}H^{(\alpha)}_{lm}
+\sum_{m=k+1}^{n}H^{(\alpha)}_{jm}\right]}\\
&=
\sum_{j=0}^{15}\sum_{k=1}^{n}
\norm\Bigg{\sum_{m=k+1}^{n}\left[H^{(\alpha)}_{jk},H^{(\alpha)}_{jm}\right]}\notag\\
&=
\sum_{j=0}^{15}
\norm\Bigg{\sum_{m=2}^{n}\left[H^{(\alpha)}_{j1},H^{(\alpha)}_{jm}\right]}
\le
\sum_{j=0}^{15}\sum_{m=2}^{n}
\norm*{\left[H^{(\alpha)}_{j1},H^{(\alpha)}_{jm}\right]}\notag\\
&=
\frac{n-1}{4h^2}\sum_{j=0}^{15}\abs*{\lambda_j^{(\alpha)}}^2
=
\frac{n-1}{4h^2}
\trnormBigg{\left(B_{\mathrm{cell}}^{-1/2}A^{(\alpha)}B_{\mathrm{cell}}^{-1/2}\right)^2}.
\label{eq:normsum1}
\end{align}
Next, we evaluate $\Xi_{\alpha\beta}^{\mathrm{cross}}$ for $\alpha\neq \beta$.
Since $S_k^{(\alpha)}$ and $S_m^{(\beta)}$ act on different axis registers, they commute: $[S_k^{(\alpha)},S_m^{(\beta)}]=0$.
Hence
\begin{align}
\left[H^{(\alpha)}_{jk},H^{(\beta)}_{lm}\right]
=
-\frac{1}{4h^2}
\left[\lambda_j^{(\alpha)}\ketbra*{\phi_j^{(\alpha)}},
\lambda_l^{(\beta)}\ketbra*{\phi_l^{(\beta)}}\right]
\otimes
S_k^{(\alpha)}S_m^{(\beta)}.
\end{align}
So
\begin{align}
\Xi_{\alpha\beta}^{\mathrm{cross}}
&=
\sum_{j=0}^{15}\sum_{k=1}^{n}
\norm\Bigg{\sum_{l=0}^{15}\sum_{m=1}^{n}\left[H^{(\alpha)}_{jk},H^{(\beta)}_{lm}\right]}\notag\\
&=
\sum_{j=0}^{15}\frac{1}{4h^2}\sum_{k=1}^{n}
\norm\Bigg{
\left[\lambda_j^{(\alpha)}\ketbra*{\phi_j^{(\alpha)}},
B_{\mathrm{cell}}^{-1/2}A^{(\beta)}B_{\mathrm{cell}}^{-1/2}\right]
\otimes
S_k^{(\alpha)}\!\left(\sum_{m=1}^{n}S_m^{(\beta)}\right)
}\notag\\
&=
\sum_{j=0}^{15}\frac{1}{4h^2}\sum_{k=1}^{n}
\norm\Bigg{
\left[\lambda_j^{(\alpha)}\ketbra*{\phi_j^{(\alpha)}},
B_{\mathrm{cell}}^{-1/2}A^{(\beta)}B_{\mathrm{cell}}^{-1/2}\right]}
\norm*{S_k^{(\alpha)}}\norm*{\sum_{m=1}^{n}S_m^{(\beta)}}\notag\\
&\le
\sum_{j=0}^{15}\sum_{k=1}^{n}\frac{1}{2h^2}
\norm\Bigg{
\left[\lambda_j^{(\alpha)}\ketbra*{\phi_j^{(\alpha)}},
B_{\mathrm{cell}}^{-1/2}A^{(\beta)}B_{\mathrm{cell}}^{-1/2}\right]}\notag\\
&\le
\sum_{j=0}^{15}\sum_{k=1}^{n}\frac{1}{h^2}\abs*{\lambda_j^{(\alpha)}}
\norm*{B_{\mathrm{cell}}^{-1/2}A^{(\beta)}B_{\mathrm{cell}}^{-1/2}}\notag\\
&=
\frac{n}{h^2}
\trnormBigg{B_{\mathrm{cell}}^{-1/2}A^{(\alpha)}B_{\mathrm{cell}}^{-1/2}}
\norm*{B_{\mathrm{cell}}^{-1/2}A^{(\beta)}B_{\mathrm{cell}}^{-1/2}}.
\label{eq:normsum2}
\end{align}
In the second equality, we use Eq.~\ref{eq:eigendecomposition}.
Also, we used
\(
\norm*{S_k^{(\alpha)}}=1
\)
(Ref.~\cite{Hamiltoniansimulationhyperbolicpartialdifferentialequationsscalablequantumcircuits2024}, Eq.~(34)),
\(
\norm*{\sum_{m=1}^{n}S_m^{(\beta)}}\le 2
\)
(Ref.~\cite{LectureNotesQuantumAlgorithmsScientificComputation2022}, Proposition~4.12),
and
\(
\norm*{[X,Y]}\le 2\norm*{X}\norm*{Y}
\).

Applying Lemma~\ref{lemma:BAB-bounded} to Eqs.~\eqref{eq:normsum1} and \eqref{eq:normsum2},
\begin{align}
\Xi_{\alpha}^{\mathrm{same}}
&\le
\frac{n-1}{4h^2}\cdot
6\rho^{-1}\norm*{\mathcal{S}_{\mathrm{comp}}^{-1}},\\
\Xi_{\alpha\beta}^{\mathrm{cross}}
&\le
\frac{n}{h^2}\cdot
6\rho^{-1}\norm*{\mathcal{S}_{\mathrm{comp}}^{-1}}.
\end{align}
Substituting these into Eq.~\eqref{eq:t-error} gives
\begin{align}
\norm*{U_1(\tau)- e^{-i\mathcal{H}\tau}}
&\le
\frac{\tau^2}{2}
\left[
3\cdot\frac{n-1}{4h^2}\cdot 6\rho^{-1}\norm*{\mathcal{S}_{\mathrm{comp}}^{-1}}
+
3\cdot\frac{n}{h^2}\cdot 6\rho^{-1}\norm*{\mathcal{S}_{\mathrm{comp}}^{-1}}
\right]\notag\\
&=
\frac{9\tau^2}{4h^2}\rho^{-1}\norm*{\mathcal{S}_{\mathrm{comp}}^{-1}}(5n-1),
\end{align}
which is Eq.~\eqref{eq:1st-trotter-error}.
\end{proof}

\section{Operator Norm of the Isotropic Compliance Matrix}\label{app:isotropic-compliance}
\begin{proposition}[Operator norm of $S_{\mathrm{comp}}^{-1}$ under isotropic assumption]
Assume a linear isotropic medium.
Let $S_{\mathrm{comp}}$ be the compliance matrix in Eq.~\eqref{eq:Scomp-of-isotropic}, i.e.,
\begin{align}
    S_{\mathrm{comp}}=\frac{1}{E}
    \begin{pmatrix}
        1 & -\nu & -\nu &0&0&0\\
        -\nu & 1 & -\nu &0&0&0\\
        -\nu & -\nu & 1 &0&0&0\\
        0&0&0& 1+\nu &0&0\\
        0&0&0& 0 &1+\nu&0\\
        0&0&0& 0 &0&1+\nu
    \end{pmatrix},
\end{align}
with $ E>0,\,-1<\nu<1/2$.
Then
\begin{align}
    \norm*{S_{\mathrm{comp}}^{-1}}
    =
    \begin{cases}
        \dfrac{E}{1-2\nu}, & 0\le \nu < \dfrac12,\\
        \dfrac{E}{1+\nu}, & -1<\nu<0.
    \end{cases}
\end{align}
\end{proposition}
\begin{proof}
For a square matrix $M$, let $\spec{M}$ denote its spectrum (set of eigenvalues).
Write
\begin{align}
S_{\mathrm{comp}}
=
\begin{bmatrix}
S_{\mathrm{comp}}^{0} & 0\\
0 & \dfrac{1+\nu}{E}I_3
\end{bmatrix},
\qquad
S_{\mathrm{comp}}^{0}
=
\frac{1+\nu}{E}I_3-\frac{\nu}{E}J,
\end{align}
where $J$ is the $3\times 3$ all ones matrix.
Since $\spec{J}=\{3,0,0\}$, we obtain
\begin{align}
\spec(S_{\mathrm{comp}}^{0})
=
\left\{
\frac{1-2\nu}{E},
\frac{1+\nu}{E},
\frac{1+\nu}{E}
\right\}.
\end{align}
Hence
\begin{align}
\spec(S_{\mathrm{comp}})
=
\left\{
\frac{1-2\nu}{E},
\underbrace{\frac{1+\nu}{E},\ldots,\frac{1+\nu}{E}}_{5\ \text{times}}
\right\},
\end{align}
So
\begin{align}
\spec(S_{\mathrm{comp}}^{-1})
=
\left\{
\frac{E}{1-2\nu},
\underbrace{\frac{E}{1+\nu},\ldots,\frac{E}{1+\nu}}_{5\ \text{times}}
\right\}.
\end{align}
Because $S_{\mathrm{comp}}^{-1}$ is real symmetric positive definite, its operator norm equals its largest eigenvalue.
Therefore
\begin{align}
\norm*{S_{\mathrm{comp}}^{-1}}
=
\begin{cases}
\dfrac{E}{1-2\nu}, & 0\le \nu < \dfrac12,\\
\dfrac{E}{1+\nu},  & -1<\nu<0.
\end{cases}
\end{align}
\end{proof}

\section{Classical time integration of the same semidiscrete system}
\label{app:classical-same-ode}

This appendix provides the details behind the classical comparison in Sec.~\ref{sec:comparison-classical-same-ode}. We work throughout on the physical nine component sector introduced there. Define
\begin{align}
\mathcal X:=\mathbb{C}^{3N^3},
\qquad
\mathcal Y:=\mathbb{C}^{6N^3},
\qquad
\mathcal Z:=\mathcal X\oplus \mathcal Y.
\end{align}
Then any transformed state vector is written as
\begin{align}
\bm{z}=
\begin{bmatrix}
\bm{q}\\
\bm{r}
\end{bmatrix}
\in \mathcal Z,
\qquad
\bm{q}\in \mathcal X,
\quad
\bm{r}\in \mathcal Y,
\end{align}
and \(L_h:\mathcal Y\to\mathcal X\) is given by Eq.~\eqref{eq:classical-Lh-main}.

\begin{proposition}[Block structure, norm identity, and norm preservation]
\label{prop:classical-Kh-app}
The operator \(K_h\) in Eq.~\eqref{eq:classical-Kh-main} is anti-Hermitian and satisfies
\begin{align}
\|K_h\|=\|L_h\|.
\label{eq:classical-Kh-Lh-same}
\end{align}
Moreover,
\begin{align}
\|L_h\|
\le
\frac{3}{h}\rho^{-1/2}\|S_{\mathrm{comp}}^{-1/2}\|
=
\frac{3v}{h}.
\label{eq:classical-Lh-bound-app}
\end{align}
Consequently, the exact semidiscrete flow preserves the norm:
\begin{align}
\|\tilde{\bm{u}}_h(t)\|=\|\tilde{\bm{u}}_h(0)\|.
\end{align}
\end{proposition}

\begin{proof}
Since \(\widetilde{A}\) is anti-Hermitian and \(\widetilde{B}\) is Hermitian positive definite,
\begin{align}
K_h=\widetilde{B}^{-1/2}\widetilde{A}\widetilde{B}^{-1/2}
\end{align}
is anti-Hermitian. From the block form
\begin{align}
K_h=
\begin{bmatrix}
O & L_h\\
-L_h^* & O
\end{bmatrix},
\end{align}
we have
\begin{align}
K_h^*K_h
=
\begin{bmatrix}
L_hL_h^* & O\\
O & L_h^*L_h
\end{bmatrix},
\end{align}
and therefore
\begin{align}
\|K_h\|^2
=
\|K_h^*K_h\|
=
\max\{\|L_hL_h^*\|,\|L_h^*L_h\|\}
=
\|L_h\|^2,
\end{align}
which proves Eq.~\eqref{eq:classical-Kh-Lh-same}. Next,
\begin{align}
\|L_h\|
&\le
\sum_{\alpha=1}^{3}
\rho^{-1/2}\|C^{(\alpha)}\|\,
\|S_{\mathrm{comp}}^{-1/2}\|\,
\|D^{(\alpha)}\|.
\end{align}
Each \(C^{(\alpha)}\) is a coordinate selection matrix with \(\|C^{(\alpha)}\|=1\), and \(\|D^{(\alpha)}\|\le 1/h\), exactly as used earlier in the proof of Proposition~\ref{thm:2nd-trotter-error}. Hence
\begin{align}
\|L_h\|
\le
\frac{3}{h}\rho^{-1/2}\|S_{\mathrm{comp}}^{-1/2}\|
=
\frac{3v}{h}.
\end{align}
Finally, since \(K_h\) is anti-Hermitian,
\begin{align}
\dv{t}\|\tilde{\bm{u}}_h(t)\|^2
=
\langle K_h\tilde{\bm{u}}_h(t),\tilde{\bm{u}}_h(t)\rangle
+
\langle \tilde{\bm{u}}_h(t),K_h\tilde{\bm{u}}_h(t)\rangle
=
0,
\end{align}
so \(\|\tilde{\bm{u}}_h(t)\|\) is conserved.
\end{proof}

\begin{proposition}[Exact subflows and one step update]
\label{prop:classical-update-app}
Let
\begin{align}
K_{1,h}
:=
\begin{bmatrix}
O & L_h\\
O & O
\end{bmatrix},
\qquad
K_{2,h}
:=
\begin{bmatrix}
O & O\\
-L_h^* & O
\end{bmatrix}.
\end{align}
Then
\begin{align}
K_{1,h}^2=K_{2,h}^2=O,
\qquad
e^{\tau K_{1,h}}=I+\tau K_{1,h},
\qquad
e^{\tau K_{2,h}}=I+\tau K_{2,h}.
\end{align}
Consequently, for \(\tilde{\bm{u}}_h^m=[\bm{q}^m;\bm{r}^m]\) and the intermediate vector
\(
\tilde{\bm{u}}_h^{m+\frac12}:=[\bm{q}^{m+\frac12};\bm{r}^m],
\)
the next step \(\tilde{\bm{u}}_h^{m+1}=\Psi_\tau \tilde{\bm{u}}_h^m\) is given by
\begin{align}
\bm{q}^{m+\frac12}
&=
\bm{q}^m+\frac{\tau}{2}L_h\bm{r}^m,
\\
\bm{r}^{m+1}
&=
\bm{r}^m-\tau L_h^*\bm{q}^{m+\frac12},
\\
\bm{q}^{m+1}
&=
\bm{q}^{m+\frac12}+\frac{\tau}{2}L_h\bm{r}^{m+1}.
\end{align}
\end{proposition}

\begin{proof}
Both \(K_{1,h}\) and \(K_{2,h}\) are strictly triangular block matrices, hence nilpotent of index two. Therefore the exponential series truncate after the linear term. Substituting these exact subflows into
\begin{align}
\Psi_\tau
=
e^{\frac{\tau}{2}K_{1,h}}
e^{\tau K_{2,h}}
e^{\frac{\tau}{2}K_{1,h}}
\end{align}
and applying the factors successively to \(\tilde{\bm{u}}_h^m=[\bm{q}^m;\bm{r}^m]\) yields the stated update formulas.
\end{proof}

\begin{proposition}[Power boundedness]
\label{prop:classical-stability-app}
Suppose
\begin{align}
\tau\|L_h\|\le \eta<2.
\label{eq:classical-stability-condition-app}
\end{align}
Then, for all \(m\in\mathbb N\),
\begin{align}
\|\Psi_\tau^m\|\le C_\eta,
\qquad
C_\eta:=\left(1-\frac{\eta^2}{4}\right)^{-1/2}.
\label{eq:classical-power-bounded-app}
\end{align}
\end{proposition}

\begin{proof}
Let
\begin{align}
L_h=\sum_{j=1}^{r}\sigma_j\,u_jv_j^*
\end{align}
be a singular value decomposition of \(L_h\) corresponding to its nonzero singular values, where
\(\sigma_j>0\), \(\{u_j\}_{j=1}^r\subset\mathcal X\) is orthonormal, and
\(\{v_j\}_{j=1}^r\subset\mathcal Y\) is orthonormal. Then
\begin{align}
\mathcal X
&=
\operatorname{span}\{u_1,\dots,u_r\}\oplus \ker L_h^*,
\\
\mathcal Y
&=
\operatorname{span}\{v_1,\dots,v_r\}\oplus \ker L_h.
\end{align}
For each \(j=1,\dots,r\), define
\begin{align}
E_j
:=
\operatorname{span}\left\{
\begin{bmatrix}
u_j\\
0
\end{bmatrix},
\begin{bmatrix}
0\\
v_j
\end{bmatrix}
\right\}
\subset \mathcal Z=\mathcal X\oplus\mathcal Y,
\end{align}
and define
\begin{align}
E_{\ker}:=\ker L_h^*\oplus\ker L_h.
\end{align}
Then
\begin{align}
\mathcal Z
=
\left(\bigoplus_{j=1}^{r}E_j\right)\oplus E_{\ker},
\end{align}
and this decomposition is orthogonal.

On \(E_{\ker}\), applying \(L_h\) and \(L_h^*\) annihilate any vector $r$ and $q$, so \(\Psi_\tau\) acts as the identity. On \(E_j\), using Proposition~\ref{prop:classical-update-app} and the relations
\begin{align}
L_h v_j=\sigma_j u_j,
\qquad
L_h^*u_j=\sigma_j v_j,
\end{align}
the restriction of \(\Psi_\tau\) to \(E_j\) is represented in the basis
\(
\left\{
\begin{bmatrix}u_j\\0\end{bmatrix},
\begin{bmatrix}0\\v_j\end{bmatrix}
\right\}
\)
by
\begin{align}
M_{\sigma_j}(\tau)
=
\begin{bmatrix}
1-\frac{\tau^2\sigma_j^2}{2}
&
\tau\sigma_j\left(1-\frac{\tau^2\sigma_j^2}{4}\right)
\\[1ex]
-\tau\sigma_j
&
1-\frac{\tau^2\sigma_j^2}{2}
\end{bmatrix}.
\label{eq:classical-Msigma-app}
\end{align}
Therefore \(\Psi_\tau\) is unitarily equivalent to
\begin{align}
\left(\bigoplus_{j=1}^{r}M_{\sigma_j}(\tau)\right)\oplus I_{E_{\ker}},
\end{align}
and hence
\begin{align}
\|\Psi_\tau^m\|
=
\max\left\{
1,\max_{1\le j\le r}\|M_{\sigma_j}(\tau)^m\|
\right\}.
\label{eq:classical-norm-reduction-app}
\end{align}

Fix \(\sigma\in\{\sigma_1,\dots,\sigma_r\}\) and define
\begin{align}
q_\sigma
:=
\sqrt{1-\frac{\tau^2\sigma^2}{4}},
\qquad
c_\sigma
:=
1-\frac{\tau^2\sigma^2}{2},
\qquad
s_\sigma
:=
\tau\sigma q_\sigma.
\end{align}
By Eq.~\eqref{eq:classical-stability-condition-app}, \(0<q_\sigma\le 1\). Moreover,
\begin{align}
c_\sigma^2+s_\sigma^2=1,
\end{align}
so
\begin{align}
Q_\sigma
:=
\begin{bmatrix}
c_\sigma & s_\sigma\\
-s_\sigma & c_\sigma
\end{bmatrix}
\end{align}
is unitary. Let
\begin{align}
S_\sigma
:=
\begin{bmatrix}
1 & 0\\
0 & q_\sigma
\end{bmatrix}.
\end{align}
A direct calculation gives
\begin{align}
M_\sigma(\tau)=S_\sigma^{-1}Q_\sigma S_\sigma.
\end{align}
Therefore,
\begin{align}
\|M_\sigma(\tau)^m\|
&=
\|S_\sigma^{-1}Q_\sigma^m S_\sigma\|
\le
\|S_\sigma^{-1}\|\,\|Q_\sigma^m\|\,\|S_\sigma\|
=
q_\sigma^{-1}.
\end{align}
Since \(\tau\sigma\le \tau\|L_h\|\le \eta\),
\begin{align}
q_\sigma^{-1}
=
\left(1-\frac{\tau^2\sigma^2}{4}\right)^{-1/2}
\le
\left(1-\frac{\eta^2}{4}\right)^{-1/2}
=
C_\eta.
\end{align}
Combining this with Eq.~\eqref{eq:classical-norm-reduction-app} proves Eq.~\eqref{eq:classical-power-bounded-app}.
\end{proof}

\begin{proposition}[One step local error]
\label{prop:classical-local-error-app}
If
\begin{align}
\tau\|L_h\|\le 1,
\end{align}
then
\begin{align}
\|e^{\tau K_h}-\Psi_\tau\|
\le
\frac{1}{2}\tau^3\|L_h\|^3.
\label{eq:classical-local-error-app}
\end{align}
\end{proposition}

\begin{proof}
By Proposition~\ref{prop:classical-update-app},
\begin{align}
\Psi_\tau
&=
\left(I+\frac{\tau}{2}K_{1,h}\right)
\left(I+\tau K_{2,h}\right)
\left(I+\frac{\tau}{2}K_{1,h}\right)
\\
&=
I+\tau K_h+\frac{\tau^2}{2}K_h^2+\frac{\tau^3}{4}K_{1,h}K_{2,h}K_{1,h},
\label{eq:classical-Psi-expansion-app}
\end{align}
since \(K_{1,h}^2=K_{2,h}^2=O\). On the other hand,
\begin{align}
e^{\tau K_h}
=
I+\tau K_h+\frac{\tau^2}{2}K_h^2+\frac{\tau^3}{6}K_h^3+\mathcal{R}_4(\tau),
\end{align}
where
\begin{align}
\|\mathcal{R}_4(\tau)\|
\le
\sum_{m=4}^{\infty}\frac{\tau^m\|K_h\|^m}{m!}
\le
\frac{\tau^4\|K_h\|^4}{24}e^{\tau\|K_h\|}.
\end{align}
Because \(K_h=K_{1,h}+K_{2,h}\) and \(K_{1,h}^2=K_{2,h}^2=O\),
\begin{align}
K_h^3=K_{1,h}K_{2,h}K_{1,h}+K_{2,h}K_{1,h}K_{2,h}.
\end{align}
Subtracting Eq.~\eqref{eq:classical-Psi-expansion-app} from the exponential expansion yields
\begin{align}
e^{\tau K_h}-\Psi_\tau
=
-\frac{\tau^3}{12}K_{1,h}K_{2,h}K_{1,h}
+
\frac{\tau^3}{6}K_{2,h}K_{1,h}K_{2,h}
+
\mathcal{R}_4(\tau).
\end{align}
Now \(\|K_{1,h}\|=\|K_{2,h}\|=\|L_h\|\) and, by Proposition~\ref{prop:classical-Kh-app}, \(\|K_h\|=\|L_h\|\). Therefore
\begin{align}
\left\|
-\frac{\tau^3}{12}K_{1,h}K_{2,h}K_{1,h}
+
\frac{\tau^3}{6}K_{2,h}K_{1,h}K_{2,h}
\right\|
\le
\frac{\tau^3}{4}\|L_h\|^3.
\end{align}
If \(\tau\|L_h\|\le 1\), then
\begin{align}
\|\mathcal{R}_4(\tau)\|
\le
\frac{e}{24}\tau^4\|L_h\|^4
\le
\frac{e}{24}\tau^3\|L_h\|^3
<
\frac{1}{4}\tau^3\|L_h\|^3.
\end{align}
Adding the two bounds proves Eq.~\eqref{eq:classical-local-error-app}.
\end{proof}

\begin{proposition}[Global operator-norm error]
\label{thm:classical-global-error-app}
Let \(T=M\tau\), and suppose
\begin{align}
\tau\|L_h\|\le \min\{1,\eta\},
\qquad
0<\eta<2.
\end{align}
Then
\begin{align}
\|e^{T K_h}-\Psi_\tau^M\|
\le
\frac{C_\eta}{2}T\tau^2\|L_h\|^3,
\label{eq:classical-global-error-app}
\end{align}
\end{proposition}

\begin{proof}
Write
\begin{align}
\tilde{\bm{u}}_h(T)-\tilde{\bm{u}}_h^M
=
\left(e^{T K_h}-\Psi_\tau^M\right)\tilde{\bm{u}}_h^0.
\end{align}
Using the telescoping identity
\begin{align}
A^M-B^M
=
\sum_{r=0}^{M-1}A^{M-1-r}(A-B)B^r
\end{align}
with \(A=e^{\tau K_h}\) and \(B=\Psi_\tau\), we obtain
\begin{align}
e^{T K_h}-\Psi_\tau^M
=
\sum_{r=0}^{M-1}
e^{(M-1-r)\tau K_h}
\left(e^{\tau K_h}-\Psi_\tau\right)
\Psi_\tau^r.
\end{align}
Since \(K_h\) is anti-Hermitian, \(e^{tK_h}\) is unitary and thus has norm one. By Proposition~\ref{prop:classical-stability-app},
\begin{align}
\|\Psi_\tau^r\|\le C_\eta,
\end{align}
and by Proposition~\ref{prop:classical-local-error-app},
\begin{align}
\|e^{\tau K_h}-\Psi_\tau\|
\le
\frac{1}{2}\tau^3\|L_h\|^3.
\end{align}
Therefore
\begin{align}
\|e^{T K_h}-\Psi_\tau^M\|
&\le
\sum_{r=0}^{M-1}
1\cdot \frac{1}{2}\tau^3\|L_h\|^3\cdot C_\eta
=
\frac{C_\eta}{2}M\tau^3\|L_h\|^3
=
\frac{C_\eta}{2}T\tau^2\|L_h\|^3.
\end{align}
This proves Eq.~\eqref{eq:classical-global-error-app}.
\end{proof}

\begin{proposition}[Classical cost for the same semidiscrete ODE]
\label{thm:classical-cost-app}
Fix a target final-time operator-norm error bound \(\epsilon\in(0,1)\). A sufficient step size condition is
\begin{align}
\tau
\le
\min\left\{
\frac{\eta}{\|L_h\|},
\sqrt{\frac{2\epsilon}{C_\eta T\|L_h\|^3}}
\right\},
\qquad
0<\eta\le 1,
\label{eq:classical-step-condition-app}
\end{align}
and hence
\begin{align}
m_{\mathrm{cl}}
=
\mathcal{O}\!\left(
\max\left\{
T\|L_h\|,
\sqrt{\frac{T^3\|L_h\|^3}{\epsilon}}
\right\}
\right).
\label{eq:classical-step-count-app}
\end{align}
Using Eq.~\eqref{eq:classical-Lh-bound-app}, this becomes
\begin{align}
m_{\mathrm{cl}}
=
\mathcal{O}\!\left(
\max\left\{
\frac{Tv}{h},
\frac{T^{3/2}v^{3/2}}{h^{3/2}\epsilon^{1/2}}
\right\}
\right).
\label{eq:classical-step-count-v-app}
\end{align}
Moreover, one step has arithmetic cost \(\mathcal{O}(N^3)\) and memory cost \(\mathcal{O}(N^3)\). Therefore
\begin{align}
W_{\mathrm{cl}}
&=
\mathcal{O}\!\left(
N^3
\max\left\{
\frac{Tv}{h},
\frac{T^{3/2}v^{3/2}}{h^{3/2}\epsilon^{1/2}}
\right\}
\right),
\label{eq:classical-total-cost-app}
\\
M_{\mathrm{cl}}
&=
\mathcal{O}(N^3).
\end{align}
\end{proposition}

\begin{proof}
Equation~\eqref{eq:classical-step-condition-app} combines the stability restriction from Proposition~\ref{prop:classical-stability-app} with the accuracy restriction from Theorem~\ref{thm:classical-global-error-app}. Equation~\eqref{eq:classical-step-count-app} follows from \(m_{\mathrm{cl}}=T/\tau\), and Eq.~\eqref{eq:classical-step-count-v-app} follows from Eq.~\eqref{eq:classical-Lh-bound-app}.

For the cost per step, each \(D^{(\alpha)}\) is a nearest neighbor finite difference operator on the three-dimensional grid, and each \(C^{(\alpha)}\) is a fixed \(3\times 6\) matrix. Hence one application of \(L_h\) or \(L_h^*\) requires \(\mathcal{O}(N^3)\) arithmetic operations. Since the update in Proposition~\ref{prop:classical-update-app} applies \(L_h\) twice and \(L_h^*\) once, together with vector additions and scalar multiplications, one time step has arithmetic cost \(\mathcal{O}(N^3)\). Multiplying by the step count gives Eq.~\eqref{eq:classical-total-cost-app}.
\end{proof}

\end{document}